\documentclass[a4paper,11pt]{jacodesmath}

\begin{document}

\title{ \textbf{Skew generalized quasi-cyclic codes over non-chain ring $F_q+vF_q$}%\thanks{supported by ...}
}

\author[1]{Kundan Suxena \thanks{kundan$\_$2321ma08@iitp.ac.in}}
\author[1]{Indibar Debnath \thanks{indibar$\_$1921ma07@iitp.ac.in}}
\author[1]{Om Prakash\thanks{om@iitp.ac.in}}
\affil[1]{Department of Mathematics, Indian Institute of Technology Patna, India}

\maketitle
% \begin{textblock*}{10cm} (8cm,3cm)
%  \begin{center}JACODESMATH  \\
%  \footnotesize Tex File For Authors  \end{center}
% \end{textblock*}
%%%

\begin{abstract}
	For a prime $p$, let $F_q$ be the finite field of order $q= p^d$. This paper presents the study on skew generalized quasi-cyclic (SGQC) codes of length $n$ over the non-chain ring $F_q+vF_q$ where $v^2=v$ and $\theta_t$ is the Galois automorphism. Here, first, we prove the dual of an SGQC code of length $n$ is also an SGQC code of the same length and derive a necessary and sufficient condition for the existence of a self-dual SGQC code. Then, we discuss the $1$-generator polynomial and the $\rho$-generator polynomial for skew generalized quasi-cyclic codes. Further, we determine the dimension and BCH type bound for the 1-generator skew generalized quasi-cyclic codes. As a by-product, with the help of MAGMA software, we provide a few examples of SGQC codes and obtain some $2$-generator SGQC codes of index $2$.

\textbf{Keywords.}  Skew cyclic codes ,  Skew generalized quasi-cyclic codes, Gray map , Generator polynomial , Idempotent generator.
\\
\textbf{2010 AMS Classification.} 94B05, 94B15, 94B60.
\end{abstract}

%----------------------------------------------------------------

\section{Introduction}
In the theory of error-correcting codes, linear codes over finite fields play a crucial role in many error-correction schemes. Even initial works of linear codes are based on binary fields, and researchers have been continuing their study on it due to the ease of their practical implementation. Later, this study extended to codes over finite fields and rings, keeping broader aspects and optimal codes in mind. In the 1990s, some studies on cyclic and self-dual cyclic codes over the ring $\mathbb{Z}_4$ have been reported in \cite{Calderbank94, Dougherty99, Pless97}, whereas \cite{Bonnecaze99} studied the same family of codes over the ring $F_2+uF_2$. In 2000, Abualrub and Siap \cite{Abualrub07} studied cyclic codes over the rings $\mathbb{Z}_2 + u\mathbb{Z}_2$, $\mathbb{Z}_2 + u\mathbb{Z}_2 + u^2\mathbb{Z}_2$, whereas Zhu et al. studied cyclic codes over $F_2 + vF_2$ \cite{Zhu10}. Later on, cyclic, quasi-cyclic, and generalized quasi-cyclic codes over finite commutative rings have been studied by introducing different Gray maps in \cite{Esmaeili09, Siap05, Prakash23, Seguin04, Ling01}. \par

On the other hand, in 2007, Boucher et al. \cite{Boucher07} introduced cyclic codes over finite noncommutative rings (skew polynomial rings) and presented $\theta$-cyclic codes over $F_q[x, \theta]$ with restrictions on their length, where $F_q$ is the finite field and $\theta$ is an automorphism over $\mathbb{F}_q$. They have obtained some codes that improved upon previously best-known linear codes. Meanwhile, in 2011, Siap et al. \cite{Siap11} studied skew cyclic codes of arbitrary length over the finite field $F_q$. After that, some other works on skew cyclic codes over rings have been seen in \cite{Seneviratne10, Gursoy14, Islam18, Boucher08, Jitman12}.\par
Recall that skew quasi-cyclic (SQC) codes of length $n$ with index $l$ over a finite field $F_q$ are linear codes where the skew cyclic shift of any codeword is again a codeword by $l$ positions is another codeword. It is noted that SQC codes of index $l=1$ are well-known skew cyclic codes. It has been shown that the class of SQC codes has a significant contribution to the class of linear codes over finite fields and rings \cite{Abualrub10, Ozen19, Ling01, Bhaintwal12}.
%using the Chinese Remainder Theorem
Later, the notion of skew generalized quasi-cyclic (SGQC) codes over finite fields was introduced by Gao et al. \cite{Gao16} and studied with the restriction that the order of the automorphism divides the length of codes. Also, based on the structural properties of SGQC codes, Abualrub et al. \cite{Abualrub18} gave some good skew $l$-GQC codes and constructed some asymmetric quantum codes over the finite field $F_4$. Recently, Seneviratne and Abualrub \cite{Seneviratne22} studied SGQC codes of arbitrary length over the finite field and obtained many new linear codes.

\par These above works inspired and gave us a gap to consider and study the skew generalized quasi-cyclic code (SGQC code) over the finite non-chain ring $F_q+vF_q$, where $v^2=v$. By considering the automorphism $\theta_t$ as $\theta_t:a+ vb\mapsto a^{p^t} + vb^{p^t}$, we establish the algebraic structure of these codes. Since the class of SGQC code is much larger than the class of SQC codes, it opens the door to looking for better codes in this class. Here, we present $1$-generator polynomial and $\rho$-generator polynomial of these codes. Further, we show that $1$-generator idempotent polynomial exists over $F_q$ and $F_q+vF_q$, respectively, for SGQC codes.
%We provide linear, cyclic, and quasi-cyclic codes from SGQC codes.
We organize our paper as follows: Section 2 recalls some known results concerning the skew polynomial ring $S[x;\theta_t]$, where $S = F_q + vF_q$, and skew cyclic code. Section 3 provides the algebraic structure of SGQC codes and their duals over the finite non-chain ring $S$. Section 4 discusses the duality of SGQC code under certain conditions on the code length, while Section 5 presents the generator for these codes. Additionally, we introduce an idempotent generator polynomial over $F_q$ and $S$ for SGQC codes and list some $2$-generator polynomial parameters over $F_3 + vF_3$, $F_4 + vF_4$, and $F_9 + vF_9$, respectively. Finally, in Section $6$, we conclude our work.

\section{Preliminaries}\label{sec 2}
Let $F_q$ be a finite field with $q$ elements where $q=p^d$ for some prime number $p$ and a positive integer $d$. Let $S= F_q+ vF_q = \{a+ vb : a , b \in F_q\}$, where $ v^2 = v$.
Thus, $S$ is a non-chain ring with $q^2$ elements and has two maximal ideals, namely, $\langle v\rangle$ and $\langle 1-v\rangle$. For more details on this ring, we refer \cite{Prakash21}. An $S$-submodule $\mathsf{C}$ of $S^n$ is called a linear code of length $n$, and elements of $\mathsf{C}$ are said to be codewords. Subsequently, the rank of code $\mathsf{C}$ is the minimum number of generators for $\mathsf{C}$, and the free rank is the rank of $\mathsf{C}$ if it is free as a module over $S$.
   We define the Gray map $\phi :S\mapsto  {F^{2}_q}$ by $\phi(a+vb) = (a,a+b)$.

              This map $\phi$ can be naturally extended from $S^{n}$ to ${F^{2n}_q}$ by \begin{align*}
              \begin{split}
                  \phi(s_1,s_2,\dots,s_n)&=(a_1,\dots ,a_n,a_1+b_1,\dots,a_n+b_n),
              \end{split}
              \end{align*} where $s_i=a_i+vb_i\in S$, for all $i=1,2,\dots,n$. The Hamming weight $w_H(c)$ is the number of nonzero entries in $c\in F^{n}_q$ and for any pair of words $c,c^{\prime} \in F^{n}_q $, the Hamming distance $d(c,c^{\prime})=w_H(c-c^{\prime})$. Also, the Lee weight denoted by $w_L(s)=w_H(\phi(s))$ for any element $s \in S$. The Lee distance is defined by $d_L(s_1, s_2) = w_L(s_1-s_2)$ for any element $s_1, s_2 \in S$. Note that the Gray map is an isometry from $S^n$ (\text{Lee distance}) to  $F_q^{2n}$ (\text{Hamming distance}) and also preserves orthogonality.\par
Now, we define some operations on linear codes, similar to that of Definition $1$ in \cite{Gursoy14}.

  \begin{definition}
     Let $\mathcal{X}$ and $\mathcal{Y}$ be two linear codes. Then the operations $ \oplus~ \text{and} ~\otimes$ are defined as
     \begin{equation*}
        \begin{split}
       \mathcal{ X}\oplus \mathcal{Y}&=\{ x+y : x\in \mathcal{ X}, y\in \mathcal{Y}\}, \text{ and}\\
        \mathcal{ X}\otimes \mathcal{Y}&=\{(x,y) : x\in \mathcal{ X}, y\in \mathcal{Y}\}.
        \end{split}
     \end{equation*}
 \end{definition}
 Let $\mathsf{C}$ be a linear code of length $n$ over $S$. Define \begin{equation*}
 \begin{split}
    \mathsf{ C}_1&:=\{c\in F^{n}_q : c+vs\in \mathsf{C} ~\text{for some}~ s \in F^{n}_q\}, \text{ and}\\
   \mathsf{C}_2&:=\{c+s\in F^{n}_q : c+vs \in \mathsf{C}\}.
 \end{split}
    \end{equation*} Clearly, $\mathsf{C}_1$ and $\mathsf{C}_2$ are linear codes over $F_q$, and from Corollary $1$ in \cite{Gursoy14}, $\mathsf{C}$ can be expressed as $\mathsf{C}=(1-v)\mathsf{C}_1\oplus v\mathsf{C}_2$.\\
 Let $\theta_t$ be an automorphism defined on $S$ by $a^{p^t}+vb^{p^t}$ where $a+vb\in S$.
%and consider the \textit{skew polynomial ring} $S[x;\theta_t]$. $S[x;\theta_t]$ as
Clearly, $\theta_t$ acts on $F_q$ as follows: \begin{align*}
        \theta_t: &F_q\mapsto F_q \\
         &a\mapsto a^{p^t}.
    \end{align*}
     \begin{definition}
Now, we consider
\begin{align*}
    S[x;\theta_t]:= \{
    a(x)=a_0 + a_1x +\dots+a_{n-1}x^{n-1}: a_i \in S \text{ for all}~i= 0,1,2,\dots,n-1\}.
\end{align*}
Then $S[x;\theta_t]$ is a ring under the usual addition of polynomials, and multiplication is defined under $(ax^i)(bx^j)=a\theta_t^i(b)x^{i+j}\text{ for all } a, b\in S$.
 \end{definition}
Clearly, $S[x;\theta_t]$ is a noncommutative ring unless $\theta_t$ is an identity automorphism. Therefore, before establishing the structure of codes, we have to specify the existence of left/right divisibility. Recall that for $a(x), b(x) \in S[x ; \theta_t]$, $a(x)$ is a right divisor of $b(x)$, if there exists $c(x) \in S[x : \theta_t] \text{ such that } b(x)= c(x)a(x)$.% under the multiplication of $S[x ; \theta_t]$.
 \begin{theorem} \cite[Theorem $2.4$]{Ozen19} \textbf{ Right Division Algorithm}: Suppose $a(x)$ and $b(x)$ are two nonzero polynomials in $S[x ; \theta_t]$ such that the leading coefficient of $b(x)$ is a unit, then there exist unique polynomials $q(x)$ and $r(x)$ such that $a(x)=q(x)b(x)+r(x)$ where $\deg r(x) < \deg a(x)$ or $r(x)=0$.\end{theorem}
     \begin{definition} \textbf{Greatest Common Right Divisor}: %Let $a(x), ~b(x) \in S[x;\theta_t]$.
     A polynomial $d(x)$ is the greatest common right divisor (gcrd) of $a(x)$ and $b(x)$ in $S[x;\theta_t]$ if $d(x)$ is a right divisor of both $a(x), ~b(x)$ and for any other right divisor $d^{\prime}(x)$ of $a(x)$ and $b(x)$, $d^{\prime}(x)$ is a right divisor of $d(x)$. %It is denoted by $d(x)= \mathrm{gcrd}(a(x),b(x))$.
     \end{definition}
Similarly, we can define the greatest common left divisor. Obviously, to construct the skew generalized quasi-cyclic codes, we shall first look at the structure of skew cyclic codes over $S$. Hence, we must develop improved versions of the results obtained in \cite{Gursoy14}.

  \begin{definition}
 Suppose $\mathsf{C}$ is a subset of $S^n$, then $\mathsf{C}$ is said to be a \textit{skew cyclic code} of length $n$ if $\mathsf{C}$ satisfies the following:
 \begin{enumerate}
     \item $\mathsf{C}$ is an $S$-submodule of $S^n$;
     \item $\sigma(c)=(\theta_t(c_{n-1}),\theta_t(c_0),\dots,\theta_t(c_{n-2})) \in \mathsf{C}$ whenever $c=(c_0,c_1,\dots,c_{n-1}) \in \mathsf{C}$. Here, $\sigma$ is called the skew cyclic shift operator.
 \end{enumerate}
 \end{definition}
Let $S_n =\dfrac{S[x;\theta_t]}{\langle x^n-1 \rangle}$, $s(x)+(x^n-1)$ be an element of $S_n$, and $a(x)\in S[x;\theta_t]$. Define multiplication from left as \begin{equation}\label{eq.1}
    a(x)*(s(x)+\langle x^n-1 \rangle)=a(x)*s(x)+\langle x^n-1 \rangle \text{ for all }a(x) \in S[x;\theta_t].
\end{equation} Clearly, multiplication on $S_n$ is well defined.
% by the elements of $S[x;\theta_t]$.
Under operation defined in Equation (\ref{eq.1}), $S_n$ is a left $S[x;\theta_t]$-module. Also, the skew cyclic codes in $S[x;\theta_t]$ is a left $S[x;\theta_t]$-submodule of the left module $S_n$. Now, we recall some results on the skew cyclic code from \cite{Gursoy14} that we will use further.
 \begin{theorem} \label{remII 2}  \cite[Theorem $5$]{Gursoy14} Let $\mathsf{C}_1$ and $\mathsf{C}_2$ be skew cyclic codes over the field $F_q$. If $\mathsf{C}=(1-v)\mathsf{C}_1\oplus v\mathsf{C}_2$ is a skew cyclic code of length $n$ over $S$, then $\mathsf{C}=\langle a(x) \rangle,$ where $a(x)=(1-v)a_1(x)+va_2(x)$ with $a_1(x)~ \text{and}~ a_2(x)$ are generator polynomials of $\mathsf{C}_1$ and $\mathsf{C}_2$ over $F_q$, respectively, while $a(x)$ is a right divisor of $x^n-1$.
\end{theorem}
\begin{corollary} \label{remII 1}\cite[Corollary $6$]{Gursoy14} Every left submodule of $S_n$ is principally generated.\end{corollary}
   \par Let $v_1(x) \text{ and } v_2(x)$ be two polynomials in $S[x;\theta_t]$. Then $v_1(x) \text{ and } v_2(x)$ are called right coprime if there exist polynomials $u_1(x) \text{ and } u_2(x)$ in $S[x;\theta_t]$ such that $u_1(x)v_1(x) +u_2(x)v_2(x)=1$. The left coprime can be defined similarly. The following lemma presents an alternative generator set.
\begin{lemma}\label{lem 1} \cite[Lemma $2.6$]{Ozen19}
           Suppose $\mathsf{C}$ is a skew cyclic code of length $n$ over $S$ with $\mathsf{C}=\langle a(x)\rangle$ such that $x^n-1=a^{\prime}(x)a(x)$. Then, any generator of $\mathsf{C}$ can be written in the form $\mathsf{C}=\langle v(x)a(x)\rangle$ and $a^{\prime}(x) \text{and}~ v(x)$ are right coprime.
 \end{lemma}

 \section{Algebraic Structure of SGQC Codes }
In this segment, we study the structural properties of skew generalized quasi-cyclic codes. Now, we generalize the Definition $2$  of \cite{Seneviratne22} over rings. Towards this, first recall the definition of the skew generalized quasi-cyclic codes.
 %\begin{definition}
     %  Suppose $(a_1,a_2,\dots,a_{t_i}) \in S^{t_i}$ where $t_i$ is a positive integer for all $i=1,2,\dots,l$. Define the skew cyclic shift operator
       %\begin{align*}
        %   \sigma : S^{t_i} \rightarrow &~ S^{t_i} ~~ by\\
          % \sigma(a_1,a_2,\dots,a_{t_i}) =&~ (\theta_t(a_{t_i}),\theta_t(a_1),\dots,\theta_t(a_{{t_i}-1})).
%\end{align*}

%\end{definition}
\begin{definition}
    Suppose $S$ is a non-chain ring, and $\theta_t$ is an automorphism of $S$ with $|\theta_t|=m_t$. Throughout,  Let $t_1,t_2,\dots,t_l$ be positive integers and $N=t_1+t_2+\dots+t_l$. A subset $\mathsf{C}$ of $\textbf{S}=S^{t_1}\times S^{t_2}\times \dots\times S^{t_l}$ is called an SGQC code of block length $(t_1,t_2,\dots,t_l)$ and length $N$ with index $l$, if $\mathsf{C}$ satisfies the following criteria:
    \begin{enumerate}
        \item $\mathsf{C}$ is a $S$-submodule of $\textbf{S}$;
        \item if $c=(c_1,c_2,\dots,c_l),~\text{then} ~ \sigma_l(c)=(\sigma(c_1),\sigma(c_2),\dots,\sigma(c_l))\in \mathsf{C}$, where $c_i=(c_{i1},c_{i2},\dots,c_{it_{i}})\in S^{t_i}, \text{ for all } ~i=1,2,\ldots,l$.
    \end{enumerate}
\end{definition}
Hence, SGQC codes of length $N$ with index $l$ over $S$ are closed under the shift $\sigma_l$. If each $t_i's$ is equal, then SGQC codes are  skew quasi-cyclic codes over $S$. If we take $l=1$, SGQC codes are skew cyclic codes over $S$. \par
Let $a=(a_1,a_2,\dots,a_l) \in \textbf{S}, \text{where}~ a_j=(a_{j,0}+va^{\prime}_{j,0},a_{j,1}+va^{\prime}_{j,1},\dots,a_{j,t_j-1}+va^{\prime}_{j,t_j-1}) ~\text{for } j=1,2,\dots,l $. For any vector $a_j\in S^{t_j}$, we modulate the vector to the polynomial $a_j(x)=(a_{j,0}+va^{\prime}_{j,0})+(a_{j,1}+va^{\prime}_{j,1})x+\dots+(a_{j,t_j-1}+va^{\prime}_{j,t_j-1})x^{t_j-1} = \sum_{i=0}^{t_j-1} (a_{j,i}+va^{\prime}_{j,i})x^{t_j-1} $ in the left $S[x;\theta_t]$-module $S_{t_j}=\dfrac{S[x;\theta_t]}{\langle x^{t_j}-1\rangle}$.
\par
We can say that the ring $\mathbf{S^\prime}=S_{t_1}\times S_{t_2}\times\dots \times S_{t_l} $ is a left $S[x;\theta_t]$-module with left multiplication defined by
 \begin{equation} \label{eq 2}
     s(x).a(x) =(s(x).a_1(x),s(x).a_2(x),\dots,s(x).a_l(x) ),
\end{equation}
for $s(x)\in S[x;\theta_t]$ and $a(x)=(a_1(x),a_2(x),\dots,a_l(x)) \in \mathbf{S^{\prime}}$.\\
Suppose $a=(a_1,a_2,\dots,a_l)$ is an element of $\textbf{S}$. Then, the map
 \begin{align*}
    \mu&:\textbf{S}\rightarrow \mathbf{S^{\prime}}; \text{ defined by}\\
 \mu(a)= &(a_1(x),a_2(x),\dots,a_l(x))=a(x).
\end{align*}
 It defines a one-to-one correspondence. Hence, a codeword $c=(c_1,c_2,\dots,c_l) \in \mathsf{C}$ will be the form of a polynomial $c(x)=(c_1(x),c_2(x),\dots,c_l(x))$ in the set $\textbf{S}^{\prime}$.
\begin{lemma} \label{lem 2}
    $\mathsf{C}$ is an SGQC code of block length $(t_1,t_2,\dots,t_l)$ and length $N$ with index $l$ if and only if $\mathsf{C}$ is a left $S[x;\theta_t]$-submodule of $\textbf{S}^{\prime}$.
\end{lemma}
      Lemma 2 is valid for any length $N$ and $m_t$ where $|\theta_t|=m_t$. Hence, we do not need to impose the condition $m_t/n$.
\begin{theorem} \label{thm 1}
    Let $\mathsf{C}$ be a linear code over $S$ of length $N$. If $\mathsf{C}=(1-v)\mathsf{C}_1\oplus v\mathsf{C}_2$, where $\mathsf{C}_1$ and $\mathsf{C}_2$ are linear codes of length $N$ over $F_q$, then $\mathsf{C}$ is an SGQC code over $S$ if and only if $\mathsf{C}_1$ and $\mathsf{C}_2$ are SGQC codes over $F_q$ of block length $(t_1,t_2,\dots,t_l)$ and length $N$ with index $l$.
\end{theorem}
\begin{proof}
Suppose $\mathsf{C}_1$ and $\mathsf{C}_2$ are SGQC codes of length $N=t_1+t_2+\dots+t_l$ with index $l$ over $F_q$. We aim to show that $\mathsf{C}$ is an SGQC over $S$. Let $s=(s_1,s_2,\dots,s_l) \in \mathsf{C}$, where each $s_i=(a_i+vb_i)\in S$. Pick $a=(a_1,a_2,\dots ,a_l)$, where $a_i=(a_{i,1},a_{i,2},\dots ,a_{i,t_i})$, and $b=(b_1,b_2,\dots,b_l$), where $b_i=(b_{i,1},b_{i,2},\dots ,b_{i,t_i})$ for all $i=1,\dots,l$. Then $a\in \mathsf{C}_1$ and $a+b\in \mathsf{C}_2$ which implies that \begin{align*}
     \sigma_l(a)&=(\sigma (a_1),\sigma (a_2),\dots,\sigma (a_l) \in \mathsf{C}_1\\
    & =(\theta_t(a_{1,t_1}),\theta_t(a_{1,1}),\dots,\theta_t(a_{1,t_1-1}),\dots, \theta_t(a_{l,t_l}),\theta_t(a_{l,1}),\dots,\theta_t(a_{l,t_l-1}))\\
    &=((a_{1,t_1})^{p^{t}},(a_{1,1})^{p^{t}},\dots,a_{1,t_1-1})^{p^{t}},\dots,(a_{l,t_l})^{p^{t}},(a_{l,1})^{p^{t}},\dots,(a_{l,t_l-1})^{p^{t}})
\end{align*} and similarly,
\begin{align*}
     \sigma_l(a+b)&=(\sigma (a_1+b_1),\sigma (a_2+b_2),\dots,\sigma (a_l+b_l)\in \mathsf{C}_2\\
    & =(\theta_t(a_{1,t_1}+b_{1,t_1}),\theta_t(a_{1,1}+b_{1,1}),\dots,\theta_t(a_{1,t_1-1}+b_{1,t_1-1}),\dots, \theta_t(a_{l,t_l}+b_{l,t_l}),\\ &~~~~~\theta_t(a_{l,1}+b_{l,1}),\dots,\theta_t(a_{l,t_l-1}+b_{l,t_l-1}))\\
    &=((a_{1,t_1})^{p^{t}}+(b_{1,t_1})^{p^{t}},(a_{1,1})^{p^{t}}+b_{1,1})^{p^{t}},\dots,a_{1,t_1-1})^{p^{t}}+b_{1,t_1-1})^{p^{t}},\dots,(a_{l,t_l})^{p^{t}}+(b_{l,t_l})^{p^{t}},\\ &~~~~~(a_{l,1})^{p^{t}}+(b_{l,1})^{p^{t}},\dots,(a_{l,t_l-1})^{p^{t}}+(b_{l,t_l-1})^{p^{t}}).
\end{align*}
Then $\sigma_l(s)=(1-v)\sigma_l(a)+v\sigma_l(a+b) \in \mathsf{C}$, which implies that $\mathsf{C}$ is an SGQC code of length $N=t_1+t_2+\dots+t_l$ with index $l$ over $S$.\\
Conversely, assume $\mathsf{C}$ is an SGQC code over $S$ of length $N=t_1+t_2+\dots+t_l$ with index $l$. For any $a=(a_1,a_2,\dots,a_l) \in \mathsf{C}_1$ and $b=(b_1,b_2,\dots,b_l) \in \mathsf{C}_2$, and if we consider $s_i=a_i+v(-a_i+b_i),~ \text{ for all }~ i=1,\dots,l$, then $s=(s_1,s_2,\dots,s_l) \in \mathsf{C}$. Since $\mathsf{C}$ is an SGQC code over $S$, we have\begin{align*}
         \sigma_l(s)&=(\sigma(s_1),\sigma(s_2),\dots,\sigma(s_l)) \in \mathsf{C}\\
         &=(\sigma(a_1)+v\sigma(-a_1+b_1),\sigma(a_2)+v\sigma(-a_2+b_2),\dots,\sigma(a_l)+v\sigma(-a_l+b_l))\\
         &=(\theta_t(a_{1,t_1})^{p^t} +v(-a_{1,t_1}+b_{1,t_1})^{p^t}),\theta_t(a_{1,1})^{p^t} +v(-a_{1,1}+b_{1,1})^{p^t}),\dots,\theta_t(a_{1,t_1-1})^{p^t} +\\&~~~~~v(-a_{1,t_1-1}+b_{1,t_1-1})^{p^t}),\dots,\theta_t(a_{l,t_l})^{p^t} +v(-a_{l,t_l}+b_{l,t_l})^{p^t}),\theta_t(a_{l,1})^{p^t} +v(-a_{l,1}+b_{l,1})^{p^t}),\\&~~~~\dots,\theta_t(a_{l,t_1-1})^{p^t} +v(-a_{l,t_l-1}+b_{l,t_l-1})^{p^t})).
\end{align*} Therefore, $\phi(\sigma_l(s))=(\sigma_l(a), \sigma_l(b))\in \mathsf{C}_1\otimes \mathsf{C}_2$, and hence $\sigma_l(a)\in \mathsf{C}_1$ and $\sigma_l(b)\in \mathsf{C}_2$ . Thus, $\mathsf{C}_1$ and $\mathsf{C}_2$ are SGQC codes over $F_q$.
\end{proof}

\section{Duality of SGQC codes}
  In this section, we present the study of the dual of an SGQC code over $S$ of length $N$ with index $l$. Before we do this, we review some definitions.
 \par Suppose $a=(a_1,a_2,\dots, a_l) ~\text{and}~ b=(b_1,b_2,\dots, b_l)$  are two elements of $\textbf{S}=S^{t_1}\times S^{t_2} \times \dots \times S^{t_l}$, where \begin{align*}
     a_i&=(a_{i,0}+va^{\prime}_{i,0},a_{i,1}+va^{\prime}_{i,1},\dots ,a_{i,t_i-1}+va^{\prime}_{i,t_i-1}),~\text{and}\\
     b_i&=(b_{i,0}+vb^{\prime}_{i,0},b_{i,1}+vb^{\prime}_{i,1},\dots ,b_{i,t_i-1}+vb^{\prime}_{i,t_i-1}).
 \end{align*}
 The usual inner product of $a$ and $b$ defined as \begin{align*}
     \langle a,b\rangle&=\sum_{i=1}^{l}a_i.b_i \\ &=\sum_{i=1}^{l} \sum_{j=0}^{t_i-1}(a_{i,j}+va^{\prime}_{i,j}).(b_{i,j}+vb^{\prime}_{i,j}) .
 \end{align*}
        Suppose $\mathsf{C}$ is an SGQC code of block length $(t_1,t_2,\dots,t_l)$ and length $N=t_1+t_2+\dots+t_l$ with index $l$. The dual of $\mathsf{C}$ is denoted by $C^{\perp}$ and is defined as $\mathsf{C}^{\perp} =\{a=(a_1,a_2,\dots,a_l) \in \textbf{S}:\langle a, c\rangle = \sum_{i=1}^{l} a_ic_i=0, \text{ for all }~ c=(c_1,c_2,\dots,c_l)\in \mathsf{C}\}$.
\begin{theorem} \label{thm 2}
      Let $\mathsf{C}$ be an SGQC code of block length $(t_1,t_2,\dots,t_l)$ and length $N=t_1+t_2+\dots+t_l$ with index $l$ over $S$, then $\mathsf{C}^{\perp}$ is an SGQC code of length $N$ with index $l$.
\end{theorem}
\begin{proof}
    Let \begin{align*}
     s&=(s_1,s_2,\dots,s_l)\in \mathsf{C}^{\perp}\\
        &=\begin{pmatrix}(s_{1,0}+vs^{\prime}_{1,0},s_{1,1}+vs^{\prime}_{1,1},\dots,s_{1,{t_1-1}}+vs^{\prime}_{1,{t_1-1}}\\
s_{2,0}+vs^{\prime}_{2,0},s_{2,1}+vs^{\prime}_{2,1},\dots,s_{2,t_2-1}+vs^{\prime}_{2,t_2-1}\\ \dots~~~~~~~~~~~~~~ \dots~~~~~~~~\dots~~~~~~~~~~~ \dots\\ s_{l,0}+vs^{\prime}_{l,0},s_{l,1}+vs^{\prime}_{l,1},\dots,s_{l,t_{l}-1}+vs^{\prime}_{l,t_{l}-1})
        \end{pmatrix}.
    \end{align*}
    Our objective is to show $\sigma_l(s) \in \mathsf{C}^{\perp}$.\\
    i.e, \begin{align*}
       \sigma_l(s)&=(\sigma(s_1),\sigma(s_2),\dots,\sigma(s_l))\\
       & =\begin{pmatrix}(\theta_t(s_{1,t_1-1}+vs^{\prime}_{1,t_1-1}),\theta_t(s_{1,0}+vs^{\prime}{1,0}),\dots,\theta_t(s_{1,t_1-2}+vs^{\prime}_{1,t_1-2})\\
       \theta_t(s_{2,t_2-1}+vs^{\prime}_{2,t_2-1}),\theta_t(s_{2,0}+vs^{\prime}_{2,0}),\dots,\theta_t(s_{2,t_2-2}+vs^{\prime}_{2,t_2-2})\\ \dots~~~~~~~~~~~~~~~~~~~~\dots~~~~~~~~~~~~~~~~~~~~\dots~~~~~~~~~~~~~~~~\dots\\
       \theta_t(s_{l,t_l-1}+vs^{\prime}_{l,t_l-1}),\theta_t(s_{l,0}+vs^{\prime}_{l,0}),\dots,\theta_t(s_{l,t_l-2}+vs^{\prime}_{l,t_l-2})
            \end{pmatrix}.
  \end{align*}
          Now, for any $c= (c_1,c_2,\dots,c_l) \in \mathsf{C}$,$\text{ for all } i=1,2,\dots,l$ \begin{align*}
              c_i&=\begin{pmatrix}(c_{1,0}+vc^{\prime}_{1,0}, c_{1,1}+vc^{\prime}_{1,1},\dots,c_{1,t_1-1}+vc^{\prime}_{1,t_1-1}\\c_{2,0}+vc^{\prime}_{2,0},c_{2,1}+c^{\prime}_{2,1},\dots,c_{2,t_2-1}+vc^{\prime}_{2,t_2-1}\\\dots~~~~~~~~~~~\dots~~~~~~~~~~~~\dots~~~~~~~~~~~~\dots\\
                  c_{l,0}+vc^{\prime}_{l,0},c_{l,1}+vc^{\prime}_{l,1},\dots,c_{l,t_l-1}+vc^{\prime}_{l,t_l-1})
              \end{pmatrix}.
          \end{align*}
    We have to show that $\langle \sigma_l(s),c\rangle =0$.\\
    Now, \begin{align*}
        \langle \sigma_l(s),c\rangle =&\begin{pmatrix}
            \theta_t(s_{1,t_1-1}+vs^{\prime}_{1,t_1-1}).(c_{1,0}+vc^{\prime}_{1,0}) + \theta_t(s_{1,0}+vs^{\prime}_{1,0}).(c_{1,1}+vc^{\prime}_{1,1})+ \dots+\\ \theta_t(s_{1,t_1-2}+vs^{\prime}_{1,t_1-2}).(c_{1,t_1-1}+vc^{\prime}_{1,t_1-1})+\dots+\theta_t(s_{l,t_l-1}+vs^{\prime}_{l,t_l-1}).\\(c_{l,0}+vc^{\prime}_{l,0})+\dots+\theta_t(s_{l,t_l-2}+vs^{\prime}_{l,t_l-2}).(c_{l,t_l-1}+vc^{\prime}_{l,t_l-1})
        \end{pmatrix}.
    \end{align*}
    %put here equation number.
 As $\mathsf{C}$ is an SGQC code, then $\sigma_l^k(c)\in \mathsf{C}$ for any positive integer $k$.
Now, suppose that\begin{align*}
    M=\ \text{lcm}(m_t,t_1,t_2,\dots,t_l),
\end{align*} where $m_t$ is the order of automorphism $\theta_t$,  then
\begin{align*}
    \sigma_l^{M}(Y)=Y,  \quad \text{for any} \quad Y \in \textbf{S}=S^{t_1}\times S^{t_2}\times\dots\times S^{t_l}.
\end{align*} Hence, $\theta_t^{M}(y)=y, ~\text{for any}~ y \in Y$, and $\theta_t^{M-1}(y)=\theta_t^{-1}(y)$.
So, \begin{align*}
    \sigma_l^{M-1}(c)=\begin{pmatrix}
        \theta_t^{-1}(c_{1,1}+vc^{\prime}_{1,1}),\theta_t^{-1}(c_{1,2}+vc^{\prime}_{1,2}),\dots,\theta_t^{-1}(c_{1,t_l-1}+vc^{\prime}_{1,t_l-1}),\theta_t^{-1}(c_{1,0}+vc^{\prime}_{1,0})\\ \dots ~~~~~~~~~~~~~~~~~~\dots ~~~~~~~~~~~~~~~~~~~~~~\dots~~~~~~~~~~~~~~~~~~~~~\dots\\\theta_t^{-1}(c_{l,1}+vc^{\prime}_{l,1}),\theta_t^{-1}(c_{l,2}+vc^{\prime}_{l,2}),\dots,\theta_t^{-1}(c_{l,t_l-1}+vc^{\prime}_{l,t_l-1}),\theta_t^{-1}(c_{l,0}+vc^{\prime}_{l,0})
         \end{pmatrix}.
\end{align*}
Now, as $\mathsf{C}$ is an SGQC code, we have \begin{align*}
    \langle s,\sigma_l^{M-1}(c)\rangle=&
    (s_{1,0}+vs^{\prime}_{1,0}.(\theta_t^{-1}(c_{1,1}+vc^{\prime}_{1,1})+\dots+(s_{1,t_1-1}+vs^{\prime}_{1,t_1-1}).\theta_t^{-1}(c_{1,0}+vc^{\prime}_{1,0})+\dots+\\&(s_{l,0}+vs^{\prime}_{l,0}).\theta_t^{-1}(c_{l,1}+vc^{\prime}_{l,1})+\dots+(s_{l,t_l-1}+vs^{\prime}_{l,t_l-1}).\theta_t^{-1}(c_{l,0}+vc^{\prime}_{l,0})=0.
\end{align*}
%Put the equation number in the above expression.
Taking $\theta_t$ to both sides of the above equation, we get
\begin{align*}
    \theta_t(s_{1,0}+vs^{\prime}_{1,0}).(c_{1,1}+vc^{\prime}_{1,1})+\dots+\theta_t(s_{1,t_1-1}+vs^{\prime}_{1,t_1-1}).(c_{1,0}+vc^{\prime}_{1,0})+\dots+\\\theta_t(s_{l,0}+vs^{\prime}_{l,0}).(c_{l,1}+vc^{\prime}_{l,1})+\dots+\theta_t(s_{l,t_l-1}+vs^{\prime}_{l,t_l-1}).(c_{l,0}+vc^{\prime}_{l,0})=0.
\end{align*}
%Put the equation number of the above expression.
We get $\langle \sigma_l(s),c\rangle=0$ on arranging the above expression. Thus, $\sigma_l(s) \in \mathsf{C}^{\perp}$. Hence, $\mathsf{C}^{\perp} $ is an SGQC code of length $N=t_1+t_2+\dots+t_l$ with index $l$.
\end{proof}
From the Theorem \ref{thm 1} and Theorem \ref{thm 2}, the following corollary follows easily.
\begin{corollary}\label{cor 1}
    $\mathsf{C}$ is a self-dual SGQC code over $S$ if and only if $\mathsf{C}_1$ and $\mathsf{C}_2$ are self-dual SGQC codes over $F_q$.
\end{corollary}

  \par In Section $3$, we have defined the map $\mu:\textbf{S}\rightarrow \textbf{S}^{\prime}$, which defines a one-to-one correspondence between SGQC codes over $S$ of length $N=t_1+t_2+ \dots +t_l$ with index $l$ and linear codes over $\textbf{S}^{\prime}=S_{t_1} \times S_{t_2} \times \dots \times S_{t_l} $ of length $N$. Here, we consider the above map $\mu$ as a polynomial representation of the SGQC codes, and then consider the \textbf{Hermitian} inner product. Following the definition in \cite[Section $6$]{Bhaintwal12},  define a ``conjugation" map $\psi_j$ on $S_{t_j}$ by \begin{align*}
      \psi_j(ax^i)=\theta_t^{-i}(a)x^{t_j-i};0\leq i\leq {t_j-1} \text{ and } j=1,2,\dots,l,
  \end{align*}
  and the map is extended to all elements of $S_{t_j}$ by linearty of addition.
  \begin{definition}
  Suppose $a(x)=(a_1(x),a_2(x),\dots ,a_l(x)) $, and $b(x)=(b_1(x),b_2(x),\dots ,b_l(x)) $ of $\textbf{S}^{\prime}$, then the \textit{Hermitian inner product} is defined by
  \begin{align*}
      a(x)*b(x)=\sum_{i=1}^{l} a_i(x) \cdot \psi_j (b_i(x)).
  \end{align*}
  \end{definition}
   Suppose that the order of $\theta_t$ is $m_t$ and $m_t/t_j: \text{ for all } j=1,2,\dots,l$. Therefore, $\theta_t^{m_t}=1=\theta_t^{t_j}$. The following result is the generalization of Proposition $3.2$ \cite{Ling01}.
  \begin{theorem} \label{thm 3}
      Let $u, v\in \textbf{S} $, and suppose $u(x)~ \text{and} ~v(x)$ denote the polynomial representation of $u~ \text{and} ~v$, respectively. Then $\sigma_l^{r}(u)\cdot v=0 :~\text{ for all }~ 0~\leq ~r ~\leq ~t_j-1, \text{and} ~\text{ for all }~ j=1,2,\dots,l $ if and only if $u(x)*v(x)=0$.
  \end{theorem}
  \begin{proof}
     First, we assume that $u(x)*v(x)=0$. Then
 \begin{equation*} \begin{split}
          0&=\sum_{j=1}^{l}u_j(x)\cdot \psi_j(v_j(x)),  \\
          &=\sum_{j=1}^{l}(\sum_{n=0}^{t_j-1}u_{n,j}+vu'_{n,j}x^n)\cdot \psi_j (\sum_{k=0}^{t_j-1}v_{k,j}+vv'_{k,j}x^k),
          \\
          &=\sum_{j=1}^{l}(\sum_{n=0}^{t_j-1}u_{n,j}+vu'_{n,j}x^n)\cdot (\sum_{k=0}^{t_j-1} \theta_t^{-k}(v_{k,j}+vv'_{k,j})x^{t_j-k}),\\ &=\sum_{w=0}^{t_j-1} (\sum_{j=1}^{l} \sum_{i=0}^{t_j-1}((u_{i+{w},j}+vu'_{i+{w},j}) \theta_t^{w} (v_{i,j}+vv'_{i,j})))x^{w},
          \end{split}
      \end{equation*}
      where the subscript $i+{t_j-1}$ is taken modulo $t_j$,  $\text{ for all } j=1,2,\dots,l$. Comparing the coefficients of $x^{t_j-1}$ for all $j=1,2,\dots,l$, on both sides, we get
      \begin{align*}
          0&=(\sum_{j=1}^{l} \sum_{i=0}^{t_j-1}((u_{i+{w},j}+vu'_{i+{w},j}) \theta_t^{w} (v_{i,j}+vv'_{i,j})), for ~all~ 0 \leq w \leq {t_j-1},\\
          &=\theta_t^{w}(\sigma_l^{n_j-w}(u).v) , \text{ for all } j=1,2,\dots,l.
      \end{align*}
      Hence, it implies that $\sigma_l^{n_j-w}(u).v=0$ for all $0\leq w \leq t_j-1$ which is equivalent to $\sigma_l^{r}(u)\cdot v=0 :~\text{ for all }~ 0~\leq ~r ~\leq ~t_j-1$.
  \end{proof}
       \begin{corollary} \label{cor 2}
            Suppose $\mathsf{C}$ is an SGQC code of block length $(t_1,t_2,\dots,t_l)$ and length $N$ with index $l$ over $S$. Then
           \begin{align*}
               \mathsf{C}^{\perp}=\{v(x)\in \textbf{S}^{\prime} : u(x)*v(x)=0, \text{ for all } u(x)\in \mathsf{C}\}.
           \end{align*}
       \end{corollary}
  \begin{theorem} \label{thm 4}
      Let $\mathsf{C}$ be an SGQC code of block length $(t_1,t_2,\dots,t_l)$ and length $N$ with index $l$ over $S$. Then $\mu(\mathsf{C}^{\perp})=\mu(\mathsf{C})^{\perp}$, where the dual in $\textbf{S} $ and $\textbf{S}^{\prime} $ is obtained concerning Euclidean and Hermitian inner product, respectively.
  \end{theorem}
  \begin{proof}
      Suppose $\mathsf{C}$ is an SGQC code and $u\in \mathsf{C}$, then $\sigma_l^{k}(u)\in \mathsf{C} $. From Theorem \ref{thm 3}, we have $\sigma_l^{k}(u)\cdot v=0 ;0\leq k\leq t_j-1, \text{ for all } j=1,2,\dots,l$. which implies that $\mu(v)\in \mu(\mathsf{C}^{\perp})$. Now, once again from Theorem \ref{thm 3}, if $\sigma_l^{k}(u)\cdot v=0$, then \begin{align*}
          u(x)*v(x)=0, \text{ and}\\
          \mu(u)*\mu(v)=0.
      \end{align*}
      As $\mu(u) \in \mu(\mathsf{C})$, we get $\mu(v)\in \mu(\mathsf{C})^{\perp}$. Therefore, $\mu(\mathsf{C}^{\perp}) \subseteq \mu(\mathsf{C})^{\perp} $.\par
      On the contrary, assume $v(x)=\mu(v) \in \mu(\mathsf{C})^{\perp}$, then there exists $u(x)=\mu(u)$ in $\mu(\mathsf{C})$ such that $u(x)*v(x)=0=\mu(u)*\mu(v)$. From previous Theorem \ref{thm 3}, $\sigma_l^{k}(u)\cdot v=0 $. As $u\in \mathsf{C}$ and $\mathsf{C}$ is an SGQC code, then $\sigma_l^{k}(u)\in \mathsf{C}$. Therefore, $v\in \mathsf{C}$ that means $\mu(v)\in \mu(\mathsf{C}^{\perp})$. Thus, $\mu(\mathsf{C})\subseteq\mu(\mathsf{C})^{\perp}$ and as a result, we get
      $\mu(\mathsf{C})= \mu(\mathsf{C})^{\perp}$.
  \end{proof}
  \begin{corollary} \label{cor 3}
      Suppose $\mathsf{C}$ is an SGQC code of block length $(t_1,t_2,\dots,t_l)$ and length $N$ with index $l$ over $S$. Then $\mathsf{C}$ is self-dual over $S$ under the Euclidean inner product if and only if $\mu(\mathsf{C})$ is self-dual over $S_{t_j} \text{ for all } j=1,2,\dots,l$, under Hermitian inner product.
  \end{corollary}

\section{Generator set for SGQC codes}
  In this section, we will extend the results derived in Section $4$ of Gao et al. \cite{Gao16} and discuss the generator set for SGQC codes. We will also derive a BCH type bound on their minimum Hamming distance. Moreover, we will extend the results derived in Section $3$ of Seneviratne and Abualrub \cite{Seneviratne22} to SGQC codes over $S$.  \par A skew generalized quasi-cyclic code $\mathsf{C}$ of block length $(t_1,t_2,\dots,t_l)$ and length $N$ with index $l$ over $S$ is called $\rho$-generator code if $\rho$ is the smallest positive integer for which there are codewords $a_j(x)=(a_{j1}(x),a_{j2}(x),\dots,a_{jl}(x))$, where $1\leq j\leq \rho$, in $\mathsf{C}$ such that $\mathsf{C}=s_1(x)a_1(x)+s_2(x)a_2(x)+\dots+s_{\rho}(x)a_{\rho}(x)$ for some $s_1(x),s_2(x),\dots,s_{\rho}(x)$ in $S[x;\theta_t]$. The set $\{a_1(x),a_2(x),\dots,a_{\rho}(x)\}$ is called a generating set for SGQC code $\mathsf{C}$.
\subsection{1-Generator polynomial of SGQC codes over $S$}
 \begin{definition}
     If an SGQC code $\mathsf{C}$ generated by a single element $e(x)=(e_1(x),e_2(x),\dots,e_l(x))$, where $e_i(x) \in S_{t_i}, ~ \text{ for all } ~ i=1,2,\dots,l $, then $\mathsf{C}$ is called a 1-generator skew generalized quasi-cyclic code. Clearly, it has the form $\mathsf{C}= \{a(x)e(x)=(a(x)e_1(x),a(x)e_2(x),\dots,a(x)e_l(x)):a(x)\in S[x;\theta_t]\} $. The monic polynomial $f(x)$ of minimum degree satisfying $e(x)f(x)=0$ is called the parity check polynomial of $\mathsf{C}$.
 \end{definition}
 \par Suppose $\mathsf{C}$ is an SGQC code of $1$-generator of block length $(t_1,t_2,\dots,t_l)$ and length $N=t_1+t_2+\dots+t_l$ with the generator polynomial $e(x)=(e_1(x),e_2(x),\dots,e_l(x)) $, where $e_i(x)\in S_{t_i}=\dfrac{S[x;\theta_t]}{\langle x^{t_i}-1 \rangle},and ~i=1,2,\dots, l$. Define the map \begin{align*}
     \Psi_i:&\textbf{S}^{\prime} \rightarrow S_{t_i} ,\text{ given by }\\&
 (a_1(x),a_2(x),\dots,a_l(x)) \rightarrow a_i(x).
 \end{align*} It is a well-defined module homomorphism. Also, $\Psi_i(\mathsf{C})=\mathsf{C}_i$. Since $\mathsf{C}$ is a $1$-generator SGQC code over $S$, $\mathsf{C}$ is a left $S[x;\theta_t]$-submodule of $\textbf{S}^{\prime}$. Hence, $\mathsf{C}_i$ is a left $S$-submodule of $S_{t_i}$. i.e. $\mathsf{C}_i$ is a skew cyclic code of length $t_i$. By Theorem \ref{remII 2}, we have $\mathsf{C}_i=\langle e_i(x) \rangle$, where $e_i(x)=(1-v)e^1_{i}(x)+ve^2_i(x)$, and $e_i(x)$ is a right divisor of $x^{t_i}-1$, such that $x^{t_i}-1=e'_i(x)e_i(x) $ for $1 \leq i \leq l$. And by the Lemma \ref{lem 1}, any generator of $\mathsf{C}_i$ has the form $\langle q_i(x)e_i(x) \rangle$, where $e'_i(x) ~and ~ q_i(x)$ are right coprime.
 From the above discussion, we summarize this in the following theorem.
 \begin{theorem}\label{thm 5}
     Suppose $\mathsf{C}$ is a $1$-generator SGQC code of block length $(t_1,t_2,\dots,t_l)$ and length $N$ with index $l$ over $S$, then $1$-generator polynomial of $\mathsf{C}$ can be taken of the form \begin{align*}
         V(x)=(q_1(x)e_1(x),q_2(x)e_2(x),\dots,q_l(x)e_l(x)),
     \end{align*} where $e_i(x)$ is a right divisor of $x^{t_i}-1 $ and $\text{gcrd}(q_i(x),(x^{t_i}-1)/{e_i(x)})=1$.
 \end{theorem}% given in the paper SQC-code over $S$.
 In \cite{Gao16}, Theorem $4.2$ discusses the parity check polynomial for $1$-generator skew generalized quasi-cyclic code over the field $F_q$. We extend this over the ring $S$.
 \begin{theorem}\label{thm 6}
     Let $\mathsf{C}$ be a $1$-generartor SGQC code of block length $(t_1,t_2,\dots,t_l)$ and length $N$ with index $l$ , generated by $e(x)=(e_1(x),e_2(x),\dots,e_l(x)) \in \textbf{S}^{\prime} $. Assume $f(x)$ is the parity check polynomial of the $1$-generator SGQC code $\mathsf{C}$. Then
     \begin{enumerate}
        \item $f(x)=\text{lclm}\{ \dfrac{x^{t_i}-1}{gcld(e_{i}(x),x^{t_i}-1}\}_{i=1,2,\dots,l}  $.
         \item  As a submodule over $S$, dimension of $\mathsf{C}$ is $\deg (f(x))$.
     \end{enumerate}

 \end{theorem}
 \begin{proof}
     Let $\textbf{S}^{\prime}=S_{t_1}\times S_{t_2} \times \dots \times S_{t_l} $ be a $S[x;\theta_t]$-module under componentwise addition and scalar multiplication as defined in Equation (\ref{eq 2}). For $1\leq i \leq l$, define well defined module-homomorphism\begin{align*}
         \Psi_i:\textbf{S}^{\prime}\rightarrow S_{t_i}, \text{ given by}\\
         \psi_i(a_1(x),a_2(x),\dots,a_l(x))&=a_i(x).
     \end{align*}
     It implies that $e_i(x)=e^{1}_i(x)+ve^{2}_i(x) \in \psi_i(\mathsf{C})$, since $e(x)\in \mathsf{C}$ implies $\alpha(x)\cdot e(x)\in \mathsf{C}$. So, $\Psi_i(\mathsf{C})$ is a left-submodule of $S_{t_i}$ and, hence a skew cyclic code of length $t_i$ in $\dfrac{S[x;\theta_t]}{\langle x^{t_i}-1\rangle}$. It has parity check polynomial $f_i(x)=\dfrac{x^{t_i}-1}{gcld(e_i(x),x^{t_i}-1)}$. If $\alpha(x) \cdot e_i(x)=0~\text{where}~\alpha(x)=\alpha^{1}(x)+v\alpha^{2}(x) ~\text{and}~e_i(x)=e^{1}_i(x)+ve^{2}_(x)$, then $f_i(x)/\alpha(x)$; it's against the assumption for $f_i(x)$ has the least possible degree. Therefore, statement $1$ follows the assumption that $f(x)$ has the minimum possible degree.  \\
     In favour of statement $2$, define a map $\zeta: S[x;\theta_t]\rightarrow \textbf{S}^{\prime}$ defined by \begin{align*}
         \beta(x)\rightarrow \beta(x)(e_1(x),e_2(x),\dots,e_l(x)).
     \end{align*}
     It is a $S[x;\theta_t]$-module homomorphism with kernel $(f(x))$. Hence, $\mathsf{C}\cong \dfrac{S[x;\theta_t]}{\langle f(x)\rangle}$ and $\text{dim}\dfrac{S[x;\theta_t]}{\langle f(x)\rangle}$ $=$ $\deg(f(x))$. Thus, $\mathsf{C}$ has dimension of $\deg (f(x))$.\qed

 \end{proof}
 \begin{corollary} \label{cor 4}
     Suppose $\mathsf{C}$ is an SGQC code of block length $(t_1,t_2,\dots,t_l)$ and length $N$ with generator $c(x)=(c_1(x),c_2(x),\dots,c_l(x))$. Suppose $f_i(x):=gcld(f(x),x^{t_i}-1) $. With the notation of Theorem \ref{thm 6}, if $\Psi_i(\mathsf{C})$ has generator polynomial $u_i(x)$, then for some polynomial $v_i(x)$ dividing by $f_i(x)$, we have \begin{align}
         \dfrac{(x^{t_i}-1)}{f_i(x)}\cdot v_i(x)=u_i(x).
     \end{align}
 \end{corollary}
 \begin{proof}
     As f(x) is the parity-check polynomial of $\mathsf{C}$, then\begin{align}
         \textbf{0}=f(x)c(x)=f(x)(c_1(x),c_2(x),\dots,c_l(x)).
     \end{align}
    This implies that $f(x)\cdot c_i(x)=0 \pmod{x^{t_i}-1},\text{ for all }~1\leq i\leq l$. Therefore, $c_i(x)\in
\langle\dfrac{x^{t_i}-1}{f_i(x)}\rangle$, hence $ \psi_i(\mathsf{C})\subseteq \langle\dfrac{x^{t_i}-1}{f_i(x)} \rangle$. This implies $(3)$ if $u_i(x)$ is the generator polynomial. It is deduced from $f_i(x)\cdot u_i(x)=v_i(x)(x^{t_i}-1)$ and $f_i(x)\cdot u_i(x)=x^{t_i}-1$ that $v_i(x)/f_i(x)$. \qed
 \end{proof}
 In the following examples, we use Theorems \ref{thm 5} and \ref{thm 6}.
\begin{example}
    Let $S[x:\theta]$ be the skew polynomial ring under the Frobenius automorphism $\theta$ over $S=F_4+vF_4$.
    Consider the polynomials $e_1(x)=x^2+1$ a right divisor of $x^4-1$, and $e_2(x)=x^3+1$ a right divisor of $x^6-1$ in $S[x:\theta]$. Let $\mathsf{C}$ be a $1$-generator skew generalized quasi-cyclic code of block length $(4,6)$ and length $10$ with index $2$, generated by $e(x)=(e_1(x), e_2(x))$. By using theorem $6$, we compute \begin{align*}
        f(x)=&lclm \begin{Bmatrix}
            x^2+1, x^3+1
        \end{Bmatrix},\\
        =& x^4+x^3+x+1.
    \end{align*} The dimension of $\mathsf{C}$ is $4$, and its generator matrix over $S$ is given as \[ \left(\begin{array}{cccccccccc}
        1& 0&1&0&1&0&0&1&0&0 \\
        0&1&0&1&0&1&0&0&1&0\\
        1&0&1&0&0&0&1&0&0&1\\
        0&1&0&1&1&0&0&1&0&0
   \end{array} \right).\] Its Gray image is $[20,4,8]$ quasi-cyclic code of degree $10$ linear code over $F_4$.
\end{example}
\begin{example}
    Consider $S[x:\theta]$, where $\theta$ is a Frobenius automorphism over $S=F_9+vF_9$. Take the polynomials  $e_1(x)=x^2+2$ as the right divisor of $x^4-1$, and  $e_2(x)=x^6+2x^4+x^2+2$ a right divisor of $x^8-1$ in $S[x:\theta]$. Let $\mathsf{C}$ be a $1$-generator SGQC code of block length $(4,8)$ and length $12$ with index $2$, generated by $e(x)=(e_1(x), e_2(x))$. Using Theorem $6$, we obtain $f(x)=x^2+1$. The dimension of $\mathsf{C}$ is $2$, and its generator matrix over $F_9+vF_9$ is  given as \[ \left(
\begin{array}{cccccccccccc}
2 & 0 & 1 & 2 & 0 & 1 & 2 & 0 & 1 & 2 & 0 & 1 \\
0 & 2 & 0 & 2 & 0 & 2 & 0 & 2 & 0 & 2 & 0 & 2
\end{array} \right).
\]
 Its Gray image is $[24,2,12]$ quasi-cyclic code of degree $12$ linear code over $F_9$.
\end{example}
 In \cite{Gao16}, the authors provide the minimum Hamming distance bound for a $1$-generator SGQC code in the case where $|\langle \theta_t \rangle|=m_t$ divides each $t_i$ over $F_q$. Using Proposition $1$ along with Corollary $1$ of \cite{Gursoy14} and Theorem \ref{thm 2}, we give the following bound for $1$-generator skew generalized quasi-cyclic codes over $S$.

\begin{theorem} \label{thm 7}
    Let $\mathsf{C}$ be a $1$-generator skew generalized quasi-cyclic code of block length $(t_1,\dots,t_l)$ and length $N$ with index $l$ over $S$, such that $\mathsf{C}=(1-v)\mathsf{C}_1\oplus v\mathsf{C}_2$ where $\mathsf{C}_1$ and $\mathsf{C}_2$ are $1$-generator SGQC codes over $F_q$ with generator polynomial $e(x)=(e_1(x),\dots,e_l(x))$ and  $e^{\prime}(x)=(e^{\prime}_1(x),\dots,e^{\prime}_l(x))$ where $e_i(x)$ and $e^{\prime}_i(x)$ are in $ \dfrac{F_q[x;\theta_t]}{\langle x^{t_1}-1\rangle}\times \dfrac{F_q[x;\theta_t]}{\langle x^{t_2}-1\rangle}\times \dots\times \dfrac{F_q[x;\theta_t]}{\langle x^{t_l}-1\rangle}$. If $d_1$ and $d_2$ are designed distances of $\mathsf{C}_1$ and $\mathsf{C}_2$ respectively, then designed distance of $\mathsf{C}$ $d(\mathsf{C})\geq \text{min}\{d_1, d_2\}$.
\end{theorem}
We review Theorem $7$ of \cite{Gursoy14} in the below Theorem \ref{thm 8}, which gives the number of skew cyclic codes of a certain length. While following its approach, we introduce different notations suitable for our result.
\begin{theorem} \label{thm 8}
    Let $(t_i,m_t)=1$ and $x^{t_i}-1=\prod_{j=1}^{r}p_{ij}^{s_ij}(x)$ where $p_{ij}(x)\in F_q[x;\theta_t]$ is irreducible polynomial. Then the number of skew cyclic codes of length $n$ over $S$ is $\prod_{j=1}^{r}(s_{ij}+1)^2$.
\end{theorem}
With the help of the above Theorem \ref{thm 8}, we give the number of $1$-generator SGQC codes in the below theorem.
\begin{theorem} \label{thm 9}
    Suppose $\mathsf{C}$ is the $1$-generator SGQC code of block length $(t_1,\dots,t_l)$ and length $N=t_1+\dots+t_l$ with index $l$ over $S$ generated by $e(x)=(e_1(x),e_2(x),\dots,e_l(x))\in \textbf{S}^\prime$. Let $(m_t,t_i)=1$, where $|\theta_t|=m_t$ and $x^{t_i}-1=\prod_{j}^{r}v_{ij}^{s_{ij}}(x)$, where $v_{ij}(x)\in F_q[x;\theta_t]$. Then the number of $1$-generator SGQC codes of length $N$ over $S$ is $\prod_{i}^{l}(\prod_{j=1}^{r}(s_{ij}+1)^2)$.
\end{theorem}
   \begin{proof}
       Define the map \begin{align*}
            \Psi_i&:\textbf{S}^{\prime} \rightarrow S_{t_i} ,\text{ given by}\\&
 (a_1(x),a_2(x),\dots,a_l(x)) \rightarrow a_i(x).
       \end{align*}
      It is a well-defined module homomorphism. It implies that $e_i(x)=e^{1}_i(x)+ve^{2}_i(x) \in \psi_i(\mathsf{C})$, since $e(x)\in \mathsf{C}$ implies $\alpha(x)\cdot e(x)\in \mathsf{C}$. Thus, $\Psi_i(\mathsf{C})$ is a left-submodule of $S_{t_i}$ and therefore a skew cyclic code of length $t_i$ in $\dfrac{S[x;\theta_t]}{\langle x^{t_i}-1\rangle}$. Now, $x^{t_i}-1=\prod_{j=1}^{r}v^{s_{ij}}_{ij}(x)$ where $v_{ij}(x)\in F_q[x;\theta_t]$ is irreducible polynomial $\text{ for all } i=1,2,\dots,l$. By Theorem \ref{thm 8}, the number of skew cyclic codes of length $t_i$  over $S$ is $\prod_{j=1}^{r}(s_{ij}+1)^2$. Now, taking the cartesian product of $\psi_i(\mathsf{C}) \text{ for all } i=1,2,\dots,l$. i.e, $\psi_1(\mathsf{C})\times \psi_2(\mathsf{C})\times \dots \times \psi_l(\mathsf{C}) $ which is isomorphic to $\mathsf{C}$. Hence, the number of $1$-generator SGQC code over $S$ is  $\prod_{i}^{l}(\prod_{j=1}^{r}(s_{ij}+1)^2)$.\qed
   \end{proof}In the following result, We present the alternative type of generator set that we employ in our computation due to its simplicity.

   \begin{theorem} \label{thm 10}
    Let $\mathsf{C}=(1-v)\mathsf{C}_1\oplus v\mathsf{C}_2$ be a skew generalized quasi-cyclic code of block length $(t_1,\dots,t_l)$ and length $N=t_1+t_2+\dots+t_l$ with index $l$ over $S$. Let $f(x)$ and $g(x)$ be $1$-generator polynomials of $\mathsf{C}_1$ and $\mathsf{C}_2$ over $F_q$, respectively. Then $\mathsf{C}=\langle (1-v)f(x), vg(x)\rangle$.
\end{theorem}
\begin{proof}
     Since we have 1-generator SGQC codes \(\mathsf{C}_1 = \langle f(x) \rangle \) and \( \mathsf{C}_2 = \langle g(x) \rangle \), where
    \[
    f(x) = (f_1(x), f_2(x), \dots, f_l(x)) \quad \text{and} \quad g(x) = (g_1(x), g_2(x), \dots, g_l(x)),
    \]
    with each \( f_i(x) \) and \( g_i(x) \) being right divisors of \( x^{t_i} - 1 \), then
    \begin{align*}
       \mathsf{C} &= \left\{ a(x) = (1-v) r(x) f(x) + v r^{\prime}(x) g(x) \mid r(x), r^{\prime}(x) \in F_q[x: \theta_t] \right\}, \\
        &= \left\{ a(x) = \left( (1-v) r(x) f_1(x), \dots, (1-v) r(x) f_l(x) \right) + \left( v r^{\prime}(x) g_1(x), \dots, v r^{\prime}(x) g_l(x) \right) \right\} ,\\
        &= \left\{ a(x) = \left( (1-v) r(x) f_1(x) + v r^{\prime}(x) g_1(x), \dots, (1-v) r(x) f_l(x) + v r^{\prime}(x) g_l(x) \right) \right\}.
    \end{align*}
    Thus, \( \mathsf{C} \subseteq \langle (1-v)(f_1(x), \dots, f_l(x)), v(g_1(x), \dots, g_l(x)) \rangle = \langle (1-v) f(x), v g(x) \rangle \subseteq \textbf{S}^{\prime} \).
    On the other hand, consider
    \[
    (1-v) k(x) f(x) + v k^{\prime}(x) g(x) \in \langle (1-v) f(x), v g(x) \rangle,
    \]
    where \( k(x) = (k_1(x), k_2(x), \dots, k_l(x)) \) and \( k^{\prime}(x) = (k^{\prime}_1(x), k^{\prime}_2(x), \dots, k^{\prime}_l(x)) \in \textbf{S}^{\prime} \). Then
    \begin{align*}
        (1-v) k(x) &= \{ (1-v) k_1(x), (1-v) k_2(x), \dots, (1-v) k_l(x) \} \\
        &= \{ (1-v) r_1(x), (1-v) r_2(x), \dots, (1-v) r_l(x) \}, \\&\quad \text{for some} \, r_1(x), \dots, r_l(x) \in F_q[x, \theta_t],
    \end{align*}
    and
    \begin{align*}
        v k^{\prime}(x) &= \{ v k^{\prime}_1(x), \dots, v k^{\prime}_l(x) \} \\
        &= \{ v r^{\prime}_1(x), \dots, v r^{\prime}_l(x) \},\\& \quad \text{for some} \, r^{\prime}_1(x), \dots, r^{\prime}_l(x) \in F_q[x, \theta_t].
    \end{align*}
    Therefore, \( \langle (1-v) f(x), v g(x) \rangle \subseteq \mathsf{C} \), which implies that \(\mathsf{C} = \langle (1-v) f(x), v g(x) \rangle \). \qed
\end{proof}

\begin{theorem} \label{thm 11}
    Suppose $\mathsf{C}_1$ and $\mathsf{C}_2$ are SGQC codes over $F_q$ and $f(x), g(x)$ are $1$-generator polynomials of these codes, respectively. Let $\mathsf{C}=(1-v)\mathsf{C}_1\oplus v\mathsf{C}_2$. Then there is a unique polynomial $h(x)\in \textbf{S}^{\prime}$ such that $\mathsf{C}=\langle h(x) \rangle$, and each component of $h(x)$ is a right divisor of $x^{t_i}-1$ where $h(x)=(1-v)f(x)+vg(x)$.
\end{theorem}
\begin{proof}
    From Theorem \ref{thm 10}, we have $\mathsf{C}=\langle (1-v)f(x),vg(x)\rangle$. Let $h(x)=(1-v)f(x)+vg(x)=((1-v)f_1(x)+vg_1(x),(1-v)f_2(x)+vg_2(x),\dots,(1-v)f_l(x)+vg_l(x))$. Clearly, $\langle h(x)\rangle \subseteq\mathsf{C}$. Since, $(1-v)f(x)=(1-v)h(x)$ and $vg(x)=vh(x)$ we conclude that $\mathsf{C} \subseteq \langle h(x)=(1-v)f(x)+vg(x)\rangle$, which implies $\mathsf{C}=\langle h(x)\rangle$. We have $f(x)=(f_1(x),f_2(x),\dots,f_l(x))$ and $g(x)=(g_1(x),g_2(x),\dots,g_l(x))$ as generator polynomial, where each component of $f(x)$ and $g(x)$ are right divisor of $x^{t_i}-1$ in $F_q[x;\theta_t],~ \text{ for all } i=1,2,\dots,l$. Then for each $i$, $\exists~ r_i(x)$ and $r^{\prime}_i(x)$ in $\dfrac{F_q[x;\theta_t]}{x^{t_i}-1}$ such that $x^{t_i}-1=r_i(x)f_i(x)=r^{\prime}_i(x)g_i(x)$. Let $r(x)=(r_1(x),r_2(x),\dots,r_l(x)$ and $r^{\prime}(x)=(r^{\prime}_1(x),r^{\prime}_2(x),\dots,r^{\prime}_l(x)$. Now, consider the following expression: \begin{align*}
        &[(1-v)r(x)+vr^{\prime}(x)]h(x)\\&=[(1-v)r(x)+vr^{\prime}(x)][1-v)f(x)+vg(x)]\\
        &=[(1-v)^2r_1(x)f_1(x)+v^2r^{\prime}_1(x)g_1(x),\dots,(1-v)^2r_l(x)f_l(x)+v^2r^{\prime}_l(x)g_l(x)]\\
        &=[(1-v)r_1(x)f_1(x)+vr^{\prime}_1(x)g_1(x),\dots,(1-v)r_l(x)f_l(x)+vr^{\prime}_l(x)g_l(x)]\\
        &=[(1-v)(x^{t_1}-1)+v(x^{t_1}-1),\dots,(1-v)(x^{t_l}-1)+v(x^{t_l}-1)]\\
        &=[x^{t_1}-1,\dots,x^{t_l}-1].
    \end{align*} Thus, $(1-v)f_i(x)+vg_i(x)$is right divisor of $x^{t_i}-1, ~\text{ for all } i=1,2,\dots,l$. \qed
\end{proof}
 In the following examples, we use Theorems \ref{thm 7}, \ref{thm 10} and \ref{thm 11}.
 \begin{example}
    Consider the polynomials $x^4-1$ and $x^6-1$ are in $F_4[x:\theta]$, where $\theta$ is the Frobenius automorphism over $F_4$. The factorization of these polynomials are as follows: \begin{align*}
        x^4-1=&(x^2+x+t^2)(x^2+x+t)\\
        =&(x^2+t^2x+t)(x^2+tx+t)\\
        x^6-1=&(x^4+tx^3+tx+1)(x^2+tx+1)\\
        =&(x^3+t^2x^2+tx+1)(x^3+tx^2+tx+1).
    \end{align*}Here, $t$ is the generator of multiplicative group of $F_4$. Consider  $\mathsf{C}_1=\langle f_1(x),~ f_2(x)\rangle$  and $\mathsf{C}_2=\langle g_1(x),~ g_2(x)\rangle$ are 1-generator SGQC codes of block length $(4,6)$ and length $10$ with index $2$ over $F_4$ where $f_1(x)=x^2+x+t$ , $f_2(x)=x^2+tx+1$, $g_1(x)=x^2+tx+t$, and $g_2(x)=x^3+tx^2+tx+1$. $\mathsf{C}_1$ and $\mathsf{C}_2$ both are of equal dimension $5$. The generator matrices of $\mathsf{C}_1$ and $\mathsf{C}_2$ are $G_1=\begin{pmatrix}
        t&1&1&0&1&t&1&0&0&0\\
        0&t^2&1&1&0&1&t^2&1&0&0\\
        1&0&t&1&0&0&1&t&1&0\\
        1&1&0&t^2&0&0&0&1&t^2&1\\
        t&1&1&0&1&0&0&0&1&t
         \end{pmatrix}$ and $G_2=\begin{pmatrix}
             t&t&1&0&1&t&t&1&0&0\\
             0&t^2&t^2&1&0&1&t^2&t^2&1&0\\
             1&0&t&t&0&0&1&t&t&1\\
             t^2&1&0&t^2&1&0&0&1&t^2&t^2\\
             t&t&1&0&t&1&0&0&1&t
         \end{pmatrix}$, respectively. Then we obtain $\mathsf{C}_1$ as $[10,5,4]$ a near-optimal code  and $\mathsf{C}_2$ as code $[10,5,3]$ over $F_4$. Consider the $\mathsf{C}=(1-v)\mathsf{C}_1\oplus v\mathsf{C}_2$, a $1$-generator SGQC code of block length $(4,6)$ and length $10$ with index $2$ over $F_4+vF_4$. Then its generator matrix is $G=\begin{pmatrix}
             (1-v)G_1\\
             vG_2
         \end{pmatrix}$ whose gray image is code $[20,10,3]$ over $F_4$.
 \end{example}
 \begin{example}
     The factorization of $x^4-1$ and $x^6-1$ in $F_9[x:\theta]$ is given as follows:\begin{align*}
         x^4-1=&(x^2+t^2x+t^3)(x^2+t^6x+t)\\
         =&(x^3+t^5x^2+2x+t)(x+t^3)\\
         x^6-1=&(x^4+t^7x^3+t^3x^2+t^3x+t^2)(x^2+t^3x+t^2)\\
         =&(x^4+t^3x^3+t^3x^2+t^7x+t^2)(x^2+t^7x+t^2).
     \end{align*}  Consider $\mathsf{C}_1=\langle f_1(x),~f_2(x)\rangle$ and $\mathsf{C}_2=\langle g_1(x),~g_2(x)\rangle$ are $1$-generator SGQC codes of block length $(4,6)$ and length $10$ with index $2$ over $F_9$ having equal dimensions of $k_1=k_2=5$ where $f_1(x)=x^2+t^6x+t$, $f_2(x)=x^2+t^3x+t^2$, $g_1(x)=x+t^3$, and $g_2(x)=x^2+t^7x+t^2$. We obtain $\mathsf{C}_1$ as code $[10,5,4]$ and $\mathsf{C}_2$ as code $[10,5,4]$ over $F_9$. Consider $\mathsf{C}=(1-v)\mathsf{C}_1\oplus v\mathsf{C}_2$, a $1$-generator SGQC code of block length $(4,6)$ and length $10$ of index $2$ over $F_9+vF_9$. Then the generator polynomial of $\mathsf{C}$ is $\langle (1-v)f_1(x)+vg_1(x), ~(1-v)f_2(x)+vg_2(x)\rangle$ and the dimension and minimum distance are $k_1+k_2=10$ and $4$, respectively. Thus, the gray image of $\mathsf{C}$ is the code $[20,10,4]$ over $F_9$.
 \end{example}

   \subsection{Idempotent generators of SGQC codes over $S$}
   In the case of a commutative ring, if $(n,q)=1$, where $q=p^d;d$ is a positive integer with $p$ being a prime number, there is a unique idempotent generator for each cyclic code of length $n$ over $F_q$. Moreover, skew cyclic codes over $F_q$ have idempotent generators under some restrictions on the length of the code. In this regard, Irfan et al. \cite{Gursoy14} already identified idempotent generators of skew cyclic codes over $S$. This subsection will prove that skew generalized quasi-cyclic codes have idempotent generators under some restrictions, followed by a few examples over $F_q$ and $S$.
\begin{theorem} \label{thm 12}\cite[Theorem $6$]{Gursoy14}
    Let $f(x)\in F_q[x;\theta_t]$ be a monic right divisor of polynomial $x^n-1$ and $\mathsf{C}=\langle f(x)\rangle$. If $(n,m_t)=1$ where $m_t=|\langle \theta_t \rangle |$ and $(n,q)=1$, then there exists an idempotent polynomial $e(x)\in \dfrac{ F_q[x;\theta_t]}{x^n-1}$ such that $\mathsf{C}=\langle e(x)\rangle$.
\end{theorem}

From Lemma \ref{lem 2} and Theorem \ref{thm 11}, the following theorem identifies an idempotent generator of skew generalized quasi-cyclic code $\mathsf{C}$ over $F_q$.
\begin{theorem} \label{thm 13}
    Let $\mathsf{C}$ be a $1$-generator skew generalized quasi-cyclic code over $F_q$ of block length $(t_1,t_2,\dots,t_l)$ and length $N=t_1+t_2+\dots+t_l$ with $\mathsf{C}=\langle c(x)\rangle=\langle c_1(x),c_2(x),\dots,c_l(x)\rangle$ where $c_i(x)\in F_q[x:\theta]$ is right divisor of $x^{t_i}-1$ for all $i=1,\dots,l$. If $(t_i,q)=1, (t_i,m_t)=1, ~\text{ for all }~i=1,2,\dots,l$ where $m_t=|\langle \theta_t \rangle|$, then there exists an idempotent polynomial $e(x)=(e_1(x),e_2(x),\dots,e_l(x))\in \dfrac{F_q[x;\theta_t]}{\langle x^{t_1}-1\rangle}\times \dfrac{F_q[x;\theta_t]}{\langle x^{t_2}-1\rangle}\times \dots \times \dfrac{F_q[x;\theta_t]}{\langle x^{t_l}-1\rangle}$ such that $\mathsf{C}=\langle e(x)\rangle$ .
\end{theorem}
\begin{proof}
  Suppose $\mathsf{C}$ is a $1$-generator skew generalized quasi-cyclic code over $F_q$ of block length $(t_1,\dots,t_l)$ and length $N=t_1+t_2+\dots+t_l$ with generator polynomial $c(x)=\langle c_1(x),c_2(x),\dots,c_l(x)\rangle$, where $c_i(x)\in F_q[x:\theta]$ is a right divisor of $x^{t_i}-1$. If $(t_i,q)=1, (t_i,m)=1, ~\text{ for all }~i=1,2,\dots,l$, then define the map \begin{align*}
        \Phi_i:&\dfrac{F_q[x:\theta]}{\langle x^{t_1}-1\rangle}\times \dfrac{F_q[x:\theta]}{\langle x^{t_2}-1\rangle}\times \dots \times \dfrac{F_q[x:\theta]}{\langle x^{t_l}-1\rangle} \rightarrow{\dfrac{F_q[x:\theta]}{\langle x^{t_i}-1\rangle}} \text{ by}\\&
    (a_1(x),a_2(x),\dots,a_l(x)) \rightarrow{a_i(x)}.
    \end{align*}
    It is a well defined module homomorphism, and \textbf{$\Phi_i(\mathsf{C})=\mathsf{C}_i$}. Since $\mathsf{C}$ is a $1$-generator SGQC code over $F_q[x;\theta_t]$, $\mathsf{C}$ is a left $F_q[x;\theta_t]$-submodule of $\dfrac{F_q[x:\theta]}{\langle x^{t_1}-1\rangle}\times \dfrac{F_q[x:\theta]}{\langle x^{t_2}-1\rangle}\times \dots \times \dfrac{F_q[x:\theta]}{\langle x^{t_l}-1\rangle}$. Therefore,  $\mathsf{C}_i$ is also a left submodule of  $\dfrac{F_q[x:\theta]}{\langle x^{t_i}-1\rangle}$. From Lemma \ref{lem 2}, $\mathsf{C}_i$ is a skew-cyclic code of length $t_i$, which implies that  $\mathsf{C}_i=\langle g_i(x)\rangle$, where $g_i(x)$ is a right divisor of $x^{t_i}-1$. Now, from Theorem \ref{thm 12}, there exists an idempotent polynomial $e_i(x)\in \dfrac{F_q[x:\theta]}{\langle x^{t_i}-1\rangle}$ such that  $\mathsf{C}_i=\langle e_i(x)\rangle$. By taking cartesian product of $\Phi_i(\mathsf{C})$ where $i=1,2,\dots,l$ then $\Phi_1(\mathsf{C}) \times \Phi_2(\mathsf{C}) \times \dots \times \Phi_l(\mathsf{C})$ $\cong$ $\mathsf{C}$. Since, $\phi_i(\mathsf{C})=\langle e_i(x)\rangle$, we conclude that $\mathsf{C}\cong \langle e_1(x), e_2(x),\dots,e_l(x)\rangle$, where each $e_i(x)$ is idempotent polynomial. \qed
\end{proof}Following the above Theorem \ref{thm 12} and Lemma \ref{lem 1}, the following theorem identifies an idempotent generator of skew cyclic code $\mathsf{C}$ over $S$.
\begin{theorem}\label{thm 14}\cite[Corollary $8$]{Gursoy14}
    If $\mathsf{C}=(1-v)\mathsf{C}_1\oplus v\mathsf{C}_2$ is a skew cyclic code of length $n$ over $S$ and $(n,m_t)=1$, $(n,q)=1$, then $\mathsf{C}_i$ has an idempotent generator, say $e_i(x)$ for $i=1,2$. Moreover, $e(x)=(1-v)e_1(x)+ve_2(x)$ is an idempotent generator of $\mathsf{C}$, i.e., $\mathsf{C}=\langle e(x) \rangle$.
\end{theorem}
\begin{theorem}\label{thm 15}
    Let $\mathsf{C}$ be a $1$-generator skew generalized quasi-cyclic code over $S$ of block length $(t_1,t_2,\dots,t_l)$ and $N=t_1+t_2+\dots+t_l$ with $\mathsf{C}=\langle u(x)\rangle=\langle u_1(x),u_2(x),\dots,u_l(x)$, where $u_i(x)\in S[x;\theta_t]$ is a right divisor of $x^{t_i}-1$. If $(t_i,q)=1,\text{ and } (t_i,m_t)=1, ~\text{ for all }~i=1,2,\dots,l$, where $m_t=|\langle \theta_t \rangle|$, then there exists an idempotent polynomial $e(x)=((1-v)e_1(x)+ve^{\prime}_1(x),(1-v)e_2(x)+ve^{\prime}_2(x),\dots,(1-v)e_l(x)+ve^{\prime}_l(x))$ which is an idempotent generator of $\mathsf{C}$, i.e., $\mathsf{C}=\langle e(x) \rangle$.
\end{theorem}
\begin{proof}
    Define the map \begin{align*}
         \Psi_i:\textbf{S}^{\prime} \rightarrow S_{t_i} ,\text{ given by}\\
 (a_1(x),a_2(x),\dots,a_l(x)) \rightarrow a_i(x).
    \end{align*}
    It is a well-defined module homomorphism. Here, $\Psi_i(\mathsf{C})$ is a left-submodule of $S_{t_i}$ and hence $\Psi_i(\mathsf{C})$ is a skew cyclic code of length $t_i$ in $S_{t_i}$. Now, by Theorem \ref{thm 13}, $\Psi_i(\mathsf{C})$ has an idempotent polynomial, i.e. $\Psi_i(\mathsf{C})=\langle (1-v)e_i(x)+ve^{\prime}_i(x)$, where $e_i(x)~\text{and}~e^{\prime}_i(x) $ are the idempotent polynomial generator of the constituent $\Psi_i(\mathsf{C})$ over $F_q$. The proof is now similar to Theorem \ref{thm 14}. Hence, $\mathsf{C} \cong \langle  (1-v)e_1(x)+ve^{\prime}_1(x),  (1-v)e_2(x)+ve^{\prime}_2(x),\dots, (1-v)e_l(x)+ve^{\prime}_l(x)\rangle$. \qed
\end{proof}
\begin{example}
    Let $S=F_q[x:\theta]$ where $\theta$ is a Frobenius automorphism over $F_4$. Consider polynomial $g_1(x)=x^2+x+1$, a right divisor of $x^3-1$, and $g_2(x)=x^4+x^3+x^2+x+1$, a right divisor of $x^5-1$ in $S[x:\theta]$. In addition, $g_1(x)$ and $g_2(x)$ are idempotent polynomials in $\dfrac{F_4[x:\theta]}{\langle x^3-1\rangle}$ and $\dfrac{F_4[x:\theta]}{\langle x^5-1\rangle}$, respectively. Let $\mathsf{C}$ be a $1$-generator skew generalized quasi-cyclic code of block length $(3,5)$ and length 8 of index two generated by $\mathsf{C}=\langle g_1(x), g_2(x)\rangle$ over $F_4$. Then, from Theorem $4.2$ of  \cite{Gao16}, parity check polynomial of $\mathsf{C}$ is \begin{align*}
        f(x)=&lclm\begin{Bmatrix}
            \dfrac{x^3-1}{g_1(x)}, \dfrac{x^5-1}{g_2(x)}
        \end{Bmatrix}\\
        =& x+1.
    \end{align*}
 Thus, $\mathsf{C}$ is a skew generalized quasi-cyclic code of length $8$ of index $2$ and dimension $1$. The generator matrix for $\mathsf{C}$ is given by $\begin{pmatrix}
        1&1&1&1&1&1&1&1
    \end{pmatrix}$. Then, the code $\mathsf{C}$ as the parameters $[8,1,8]$ is a $1$-generator SGQC code over $F_4$, an optimal linear code over $F_4$.
\end{example}
\begin{example}
    Using the notation of Example $1$, consider $g_3(x)=x^4+x^2+x+1$ a right divisor of $x^7-1$ in $S[x:\theta]$ and an idempotent polynomial in $\dfrac{F_4[x:\theta]}{\langle x^7-1 \rangle}$. Let $\mathsf{C}$ be a $1$-generator skew generalized quasi-cyclic codes of block length $(3,5,7)$ and length $15$ with index $3$, generated by $\mathsf{C}=\langle g_1(x), g_2(x), g_3(x) \rangle$ over $F_4$. From Theorem $4.2$ of \cite{Gao16}, parity check polynomial of $\mathsf{C}$ is $x^4+x^3+x^2+1$. Hence, $\mathsf{C}$ is a $1$-generator skew generalized quasi-cyclic code of length $15$ of index $3$ and dimension $4$. Thus, $\mathsf{C}$ havig parameters $[15,4,4]$ is a $1$-generator SGQC code over $F_4$.
\end{example}
  \subsection{$\rho$-Generator Polynomial over $S$}
  This subsection introduces a set of $\rho$-generator polynomials for SGQC codes over the ring $S$. These generator polynomials must satisfy certain constraints. Here, We develop our method upon the approach introduced by Seneviratne et al. in \cite{Seneviratne22} for generator polynomial over $F_q$. In addition, we used these results to find the cardinality and dimension of SGQC codes and provide the parameters of the Gray images of $2$-generator SGQC codes.
     \par Suppose $\mathsf{C}$ is a skew generalized quasi-cyclic code of block length $(t_1,t_2,\dots,t_l)$ and length $N$ with index $l$. Let $a(x)=(a_1(x)+va^{\prime}_1(x),\dots,a_l(x)+va^{\prime}_l(x))$. Define the sets \begin{align*}
         K_i= \begin{Bmatrix}
              p_i(x): \text{a codeword} ~~a(x)=(a_1(x)+va{\prime}_1(x),a_2(x)+va{\prime}_2(x),\dots,p_i(x),0,0,\dots,0)\in \mathsf{C},\\
              \text{where} ~p_i(x)\in \dfrac{S[x;\theta_t]}{\langle x^{t_i}-1\rangle} ~\text{and} ~a_{i+1}(x)+va{\prime}_{i+1}(x)=a_{i+2}(x)+va{\prime}_{i+2}(x)=\dots=0
         \end{Bmatrix},
     \end{align*}
     i.e.,
     \begin{align*}
         K_1&=\{ p_1(x): \text{ a codeword} ~a(x)=(p_1(x),0,0,\dots,0)\in \mathsf{C}\},\\
         K_2&=\{p_2(x):\text{ a codeword} ~~a(x)=(a_1(x)+va^{\prime}_1(x),p_2(x),0,0,\dots,0)\in \mathsf{C}\}~ \text{and}\\
         K_l&=\{p_l(x) : \text{a codeword} ~~a(x)=(a_1(x)+va^{\prime}_1(x),a_2(x)+va^{\prime}_2(x),\dots,p_l(x))\in \mathsf{C}\}.
     \end{align*}
      As $(0,0,\dots,0)\in \mathsf{C},~ K_i$ is a non-empty set $\text{ for all } i=1,2,\dots,l$.
\begin{lemma} \label{lem 3}
    The above set $K_i$ is a left submodule of $\dfrac{S[x;\theta_t]}{\langle x^{t_i}-1\rangle}$ $\text{ for all } i=1,2,\dots,l$.
\end{lemma}
 \begin{proof}
     Suppose $p_i(x),q_i(x) \in K_i $ and $s(x)\in S[x;\theta_t]$, then there are \begin{align*}
         a(x)&=(a_1(x)+va^{\prime}_1(x),a_2(x)+va^{\prime}_2(x),\dots,p_i(x),0,\dots,0) ~\text{and}\\ ~b(x)&=(b_1(x)+vb^{\prime}_1(x),b_2(x)+vb^{\prime}_2(x),\dots,q_i(x),0,\dots,0)\in \mathsf{C}.
     \end{align*}
      As $\mathsf{C}$ is a left submodule of $S_{t_1}\times S_{t_2}\times \dots \times S_{t_l}$, $a(x)+b(x)=(a_1(x)+va^{\prime}_1(x)+(b_1(x)+vb^{\prime}_1((x)),a_2(x)+va^{\prime}_2(x)+(b_2(x)+vb^{\prime}_2((x),\dots,p_i(x)+q_i(x),0,0,\dots,0)$, and $s(x)a(x)=(s(x)(a_1(x)+va^{\prime}_1(x)),s(x)(a_2(x)+va^{\prime}_2(x)),\dots,s(x)p_i(x),0,0,\dots,0)$ are codewords in $\mathsf{C}$. Thus, $p_i(x)+q_i(x)$ and $s(x)p_i(x)$ are elements in $K_i$, which implies that $K_i$ is a left submodule of $\dfrac{S[x;\theta_t]}{\langle x^{t_i}-1\rangle}$. \qed
 \end{proof}
 \textbf{Note:} From Corollary \ref{remII 1}, each $K_i$ is principally generated, i.e., $K_i=\langle f_i(x) \rangle$ where $f_i(x)$ is a right divisor of $x^{t_i}-1$.
 \begin{lemma} \label{lem 4}
     Let $\mathsf{C}$ be an SGQC code of block length $(t_1,t_2,\dots,t_l)$ and length $N$ with index $l$ and let $a(x)=(a_1(x)+va^{\prime}_1(x),a_2(x)+va^{\prime}_2(x),\dots,a_l(x)+va^{\prime}_l(x)) \in \mathsf{C}$. Then the sets \begin{align*}
         L&=\begin{Bmatrix}
             (h_(x),h_2(x),\dots,h_{l-1}(x):~ \text{a codeword} ~(h_1(x),h_2(x),\dots,h_{l-1}(x),h_l(x)) \in \mathsf{C}
         \end{Bmatrix}   \text{and}~\\
         M&=\begin{Bmatrix}
             (h_2(x),\dots,h_{l}(x): ~\text{a codeword} ~(h_1(x),h_2(x),\dots,h_{l-1}(x),h_l(x)) \in \mathsf{C},\\ ~\text{where}~ h_i(x)=h_i(x)+vh^{\prime}_i(x)
         \end{Bmatrix}
     \end{align*}
     are left submodules of $S_{t_1}\times S_{t_2}\times \dots \times S_{t_{l-1}}$ and  $S_{t_2}\times S_{t_3}\times \dots \times S_{t_l}$, respectively.
 \end{lemma}
 \begin{proof}
     The proof is almost similar to the proof of Lemma \ref{lem 3}.\qed
 \end{proof}
 Using the above-defined notation, we give the set of $\rho$-generator polynomials of the SGQC code in the following theorem.
 \begin{theorem} \label{thm 16}
      Let $\mathsf{C}$ be an SGQC code of block length $(t_1, t_2,\dots, t_l)$ and length $N$ with index $l$. Then \begin{align*}
          \mathsf{C}=&\left\langle \begin{matrix} ((1-v)f_1(x)+vf^{\prime}_1(x),0,\dots,0),(p_{21}(x),(1-v)f_2(x)+vf^{\prime}_2(x),0,\dots,0),(p_{31}(x),p_{32}(x),\\(1-v)f_3(x)+vf^{\prime}_3(x),0,\dots,0),0,\dots,(p_{l1}(x),p_{l2}(x),\dots,p_{l(l-1)}(x),(1-v)f_{l}(x)+vf^{\prime}_{l}(x))
          \end{matrix}\right\rangle,
           \end{align*}
     where $(1-v)f_i(x)+vf^{\prime}_i(x)$ is a right divisor of $x^{t_i}-1$ $\text{ for all } j=1,2,\dots,l$.
 \end{theorem}
 \begin{proof}
     We will prove it inductively. Suppose the index $l=1$, then $\mathsf{C}=\langle f_{t_1}(x) \rangle =\langle (1-v)f_1(x)+vf^{\prime}_1(x)\rangle$, i.e., $\mathsf{C}$ is a skew cyclic code which is principally generated in $\dfrac{S[x;\theta_t]}{\langle x^{t_1}-1 \rangle}$. Assume that the statement holds for all $i<l$. Let $\mathsf{C}$ be an SGQC code of length $N$ with index $l$ and $a(x)=(a_1(x)+va^{\prime}_1(x),a_2(x)+va^{\prime}_2(x),\dots,a_l(x)+va^{\prime}_l(x))\in \mathsf{C}$, where $a_i(x)+va^{\prime}_i(x) \in \dfrac{S[x;\theta_t]}{\langle x^{t_i}-1\rangle}$. From the definition of set $K_l$, we have $a_l(x)+va^{\prime}_l(x)\in K_l=\langle (1-v)f_l(x)+vf^{\prime}_l(x)\rangle$ and $a_l(x)+va^{\prime}_l(x)=(q_l(x)+vq^{\prime}_l(x))((1-v)f_l(x)+vf^{\prime}_l(x))=q_{t_l}(x)f_{t_l}(x)$ and since $f_{t_l}(x) \in K_l$, there is a codeword $(p_{l1}(x),p_{l2}(x),\dots,p_{l(l-1)}(x),f_{t_l}(x))\in \mathsf{C}$. Thus, \begin{align*}
         a(x)&=\Bigl(a_1(x)+va^{\prime}_1(x),a_2(x)+va^{\prime}_2(x),\dots,q_{t_l}(x)f_{t_l}(x)\Bigl)\\
            &=q_{t_l}(x)\begin{pmatrix}
                (p_{l1}(x),p_{l2}(x),\dots,p_{l(l-1)}(x),f_{t_l}(x))+(a_1(x)+va^{\prime}_1(x)-q_{t_l}(x)p_{l1}(x),\\a_2(x)+va^{\prime}_2(x)-q_{t_l}(x)p_{l2}(x),\dots,a_{l-1}(x)+va^{\prime}_{l-1}(x)-q_{t_l}(x)p_{l(l-1)}(x),0)
            \end{pmatrix}.
     \end{align*}
        Since \begin{align*}
             &\begin{pmatrix}
                 q_{t_l}(x)(p_{l1}(x),p_{l2}(x),\dots,p_{l(l-1)}(x),f_{t_l}(x))
             \end{pmatrix}\in \mathsf{C}, and \\&
             \begin{pmatrix}
                 a_1(x)+va^{\prime}_1(x)-q_{t_l}(x)p_{l1}(x),a_2(x)+va^{\prime}_2(x)-q_{t_l}(x)p_{l2}(x), \dots,\\a_{l-1}(x)+va^{\prime}_{l-1}(x)-q_{t_l}(x)p_{l(l-1)}(x),0
             \end{pmatrix} \in \mathsf{C}.
        \end{align*}
     By Lemma \ref{lem 4}, $(a_1(x)+va^{\prime}_1(x)-q_{t_l}(x)p_{l1}(x),a_2(x)+va^{\prime}_2(x)-q_{t_l}(x)p_{l2}(x)$,$\dots,a_{l-1}(x)+va^{\prime}_{l-1}(x)-q_{t_l}(x)p_{l(l-1)}(x))\in L$. As $L$ is a left-submodule of $S_{t_1}\times S_{t_2}\times \dots \times S_{t_{l-1}}$ and by the inductive hypothesis, we have \begin{align*}
         L=\left\langle \begin{matrix}((1-v)f_1(x)+vf^{\prime}_1(x),0,\dots,0),(p_{21}(x),(1-v)f_2(x)+vf^{\prime}_2(x),0,\dots,0),(p_{31}(x),p_{32}(x),\\(1-v)f_3(x)+vf^{\prime}_3(x),0,\dots,0),0,\dots,(p_{(l-1)1}(x),p_{(l-1)2}(x),\dots,p_{(l-1)(l-1)}(x),f_{t_{(l-1)}}(x))
         \end{matrix}\right\rangle,
     \end{align*}
       where $f_{t_i}(x)$ is a right divisor of $x^{t_i}-1$, $\text{ for all } i=1,2,\dots,l-1$. Thus, \begin{center}
           $(a_1(x)+va^{\prime}_1(x)-q_{t_l}(x)p_{l1}(x),a_2(x)+va^{\prime}_2(x)-q_{t_l}(x)p_{l2}(x),\dots,a_{l-1}(x)+va^{\prime}_{l-1}(x)-q_{t_l}(x)p_{l(l-1)}(x))$\\
           $=$ $c_1(x)((1-v)f_1(x)+vf^{\prime}_1(x),0,\dots,0)+c_2(x)(p_{21}(x),(1-v)f_2(x)+vf^{\prime}_2(x),0,\dots,0)+\dots+c_{l-1}(x)(p_{(l-1)1}(x),p_{(l-1)2}(x),\dots,p_{(l-1)(l-2)}(x),f_{t_{l-1}}(x))$,
       \end{center}
           and \begin{align*}
               a(x)&=(a_1(x)+va^{\prime}_1(x),a_2(x)+va^{\prime}_2(x),\dots,q_{t_l}(x)f_{t_l}(x))\\
               =& \begin{pmatrix}
                   (q_{t_l}(x)(p_{l1}(x),p_{l2}(x),\dots,p_{l(l-1)}(x),f_{t_l}(x))+ (a_1(x)+va^{\prime}_1(x)-q_{t_l}(x)p_{l1}(x),\\a_2(x)+va^{\prime}_2(x)-q_{t_l}(x)p_{l2}(x), \dots,a_{l-1}(x)+va^{\prime}_{l-1}(x)-q_{t_l}(x)p_{l(l-1)}(x),0)
               \end{pmatrix} \\
               = &\begin{pmatrix}
                   (q_{t_l}(x)(p_{l1}(x),p_{l2}(x),\dots,p_{l(l-1)}(x),f_{t_l}(x))+ c_1(x)((1-v)f_1(x)+vf^{\prime}_1(x),0,\dots,0)\\+\dots+c_{l-1}(x)(p_{(l-1)1}(x),p_{(l-1)2}(x),\dots,p_{(l-1)(l-2)}(x),f_{t_{l-1}}(x))
               \end{pmatrix}.
           \end{align*}
    Therefore, \begin{align*}
        \mathsf{C}&=\left\langle \begin{matrix}
             ((1-v)f_1(x)+vf^{\prime}_1(x),0,\dots,0),(p_{21}(x),(1-v)f_2(x)+vf^{\prime}_2(x),0,\dots,0),(p_{31}(x),p_{32}(x),\\(1-v)f_3(x)+vf^{\prime}_3(x),0,\dots,0),0,\dots,(p_{l1}(x),p_{l2}(x),\dots,p_{l(l-1)}(x),(1-v)f_{l}(x)+vf^{\prime}_{l}(x))
         \end{matrix}\right\rangle,
    \end{align*}
     where $(1-v)f_i(x)+vf^{\prime}_i(x)$ is a right divisor of $x^{t_i}-1$, $\text{ for all } i=1,2,\dots,l$.\qed
 \end{proof}
 \begin{theorem}\label{thm 17}
      Let $\mathsf{C}$ be an SGQC code of block length $(t_1,t_2,\dots,t_l)$ and length $N=t_1+t_2+\dots+t_l$ with index $l$ given by  \begin{align*}
         \mathsf{C}&=\left\langle \begin{matrix}
              (f_{t_1}(x),0,\dots,0),(p_{21}(x),f_{t_2}(x),0,\dots,0),(p_{31}(x),p_{32}(x),\\
              f_{t_3}(x),0,\dots,0),\dots,(p_{l1}(x),p_{l2}(x),\dots,p_{l(l-1)}(x),(f_{t_l}(x))
          \end{matrix} \right\rangle,
     \end{align*}
     where $(1-v)f_i(x)+vf^{\prime}_i(x)=f_{t_i}(x)$ is a right divisor of $x^{t_i}-1$ $\text{ for all } i=1,2,\dots,l$. Then
     \begin{enumerate}
         \item  $\deg p_{ij}(x)<\deg f_{t_j}(x)$ ~ $\text{ for all } i=2,\dots,l,$ and $j=1,2,\dots {l-1} ~\text{with } i>j$.
         \item If $x^{t_i}-1=q_{t_i}(x)f_{t_i}(x)$, then $q_{t_i}(x)p_{(i)(i-1)}(x)\in $ $\langle f_{t_{i-1}}(x)\rangle$ and $q_{t_i}(x)p_{(i)(i-1)}(x)=s_{t_i}(x)f_{t_{i-1}}(x)$, $\text{ for all } i=2,3,\dots,l$.
     \end{enumerate}
 \end{theorem}
\begin{proof}
    The first part can be proven inductively. So, we left out the first part. For the second part, we observe that \begin{align*}
        &q_{t_i}(x)(p_{i1}(x),p_{i2}(x),\dots,p_{i(i-1)}(x),f_{t_i}(x),0,\dots,0)\\
        &=(q_{t_i}(x)p_{i1}(x),q_{t_i}(x)p_{i2}(x),\dots,q_{t_i}(x)p_{i(i-1)}(x),0,\dots,0).
    \end{align*} So, $q_{t_i}(x)p_{i(i-1)}(x)\in K_{i-1}=\langle f_{t_{i-1}}(x) \rangle$. Hence, $q_{t_i}(x)p_{i(i-1)}(x)=s_{t_i}(x)f_{t_{i-1}}(x)$  $\text{ for all } i=2,\dots,l$.  \qed
\end{proof}
In the following theorem, we use the properties of these generators to give the dimension and cardinality of these codes.
\begin{theorem}\label{thm 18}
    Let $\mathsf{C}$ be an SGQC code of block length $(t_1,t_2,\dots,t_l)$ and length $N=t_1+t_2+\dots+t_l$ with index $l$. If \begin{align*}
          \mathsf{C}=\left\langle \begin{matrix} (f_{t_1}(x),0,\dots,0),(p_{21}(x),f_{t_2}(x),0,\dots,0),(p_{31}(x),p_{32}(x),\\f_{t_3}(x),0,\dots,0),0,\dots,(p_{l1}(x),p_{l2}(x),\dots,p_{l(l-1)}(x),(f_{t_l}(x))\end{matrix}\right\rangle,
     \end{align*}
     where $(1-v)f_i(x)+vf^{\prime}_i(x)=f_{t_i}(x)$ is a right divisor of $x^{t_i}-1$, for all $i=1~\text{to}~l$. Then $\text{rank}(\mathsf{C})=\deg(q_{t_1}(x)) + \deg(q_{t_2}(x))+ ~\dots ~+ \deg(q_{t_l}(x))$ and $|\mathsf{C}|=q^{2\deg(q_{t_1}(x))}q^{2\deg(q_{t_2}(x))}\dots q^{2\deg(q_{t_l}(x))}$, with keeping same notation as in Theorem \ref{thm 17}.
\end{theorem}
\begin{proof} The proof follows by applying the Principle of Mathematical Induction on index $l$ and using Theorem \ref{thm 17} and Lemma \ref{lem 4}.\qed
\end{proof}
\begin{example} Consider the polynomial factorization in $S[x:\theta]$ where $\theta$ is a Frobenius automorphism.
    We take the polynomial $x^8-1$ over $S=F_4+vF_4$. We have the following factorizations: \begin{align*}
        x^8-1=&(x^5 + (v + t^2)*x^4 + x^3 + (v + t)*x^2 + 1)*(x^3 + (v + t)*x^2 + 1) \\
        =&(x^5 + (v + t)*x^4 + x^3 + (v + t^2)*x^2 + 1)*(x^3 + (v + t^2)*x^2 + 1)\\
        =&(x^4 + v*x^3 + (t^2*v + 1)*x^2 + x + t^2)*(x^4 + v*x^3 + (t*v + 1)*x^2 + x + t)\\
        =&(x^5 + (t^2*v + 1)*x^4 + (v + 1)*x^3 + (t^2*v + t^2)*x^2 + v*x + t)*\\&(x^3 + (t*v + 1)*x^2 + v*x + t^2)\quad\text{ and so on}.
    \end{align*}
 Next, consider the factorization of $x^6-1$ over $S=F_9+vF_9$. We have\begin{align*}
        x^6-1=&(x^4 + (t^5*v + 2)*x^3 + (t*v + 1)*x + 2)*(x^2 + (t*v + 1)*x + 1)\\
        =&(x^3 + (t^2*v + 2)*x^2 + (t*v + t)*x + 2)*(x^3 + (t^2*v + 1)*x^2 + (t*v + t)*x + 1)\\
        =&(x^3 + (2*v + t)*x^2 + (t^2*v + t^5)*x + 2)*(x^3 + (v + t^7)*x^2 + (t^2*v + t^5)*x + 1)\\
        =&(x^3 + v*x^2 + 2*v*x + 2)*(x^3 + 2*v*x^2 + 2*v*x + 1) \quad\text{ and so on}.
    \end{align*}
\end{example}
One of the main motives in coding theory is obtaining codes with better parameters or better code rates. There is a well-known table of linear codes with best-known parameters on small finite fields \cite{Grassl}. That codes table has continuously been updated with new codes appearing in the literature by different researchers. \par
In Tables \ref{Table1}, \ref{Table2} and \ref{Table3}, we present the parameters of Gray images of $2$-generator skew generalized quasi-cyclic codes of index $2$ over $S$ where we consider $q=3$, $q=9$ and $q=4$, respectively. We have considered the Frobenius automorphism for each code. We write the coefficients of the generator polynomial in ascending order of the degree of the indeterminate; for example, the polynomial $f_1(x)=x^3+(v+t)x^2+x+(v+t)$ is represented by $(v+t)1(v+t)1$.
\begin{table}[H] \centering \begin{longtable}{|c|c|c|}\caption{\label{Table1}The $2$-generator SGQC codes over $F_3+vF_3$}\\
   \hline
    $t_1, t_2, N$ & Generator polynomials  & $\phi(\mathsf{C})$\\
    \hline
    $6,1,7$ & $f_{t_1}=(2v+2)2(v+1)1$, $p_{21}=(v+1)(v+2)1$, $f_{t_2}=1$ & $[14,4,7]*$\\
    \hline
    $6,2,8$ & $f_{t_1}=(2v+2)2(v+1)1$,  $p_{21}=(v+1)(v+2)1$, $f_{t_2}=11$ & $[16,4,7]$\\
    \hline
    $6,3,9$ & $f_{t_1}=(2v+2)2(v+1)1$, $p_{21}=(v+1)(v+2)1$, $f_{t_2}=111$ & $[18,4,7]$\\
    \hline
    $6,4,10$ & $f_{t_1}=(2v+2)2(v+1)1$, $p_{21}=(v+1)(v+2)1$, $f_{t_2}=1111$ & $[20,4,7]$\\
    \hline
    $6,9,15$ & $f_{t_1}=(2v+2)2(v+1)1$, $p_{21}=(v+1)(v+2)1$, $f_{t_2}=21021021$ & $[30,5,7]$\\
    \hline
      \end{longtable}
 \end{table}
 \begin{table}[H] \centering \begin{longtable}{|c|c|c|} \caption{\label{Table2}The $2$-generator SGQC codes over $F_9+vF_9$}\\
   \hline
    $t_1, t_2, N$ & Generator polynomials   & $\phi(\mathsf{C})$\\
    \hline
    $6,2,8$ & $f_{t_1}=1(v+2)(v+2)1$, $p_{21}=1(v+1)1$, $f_{t_2}=21$ & $[16,4,6]$\\
    \hline
    $6,1,7$ & $f_{t_1}=2(v+2)(2v+1)1$, $p_{21}=1(2v+2)1$, $f_{t_2}=1$ & $[14,4,6]$\\
    \hline
    $4,2,6$ & $f_{t_1}=t^52t1$, $p_{21}=t^7t1$, $f_{t_2}=t^61$ & $[12,2,8]$\\
    \hline
    $4,6,10$ & $f_{t_1}=t^52t1$, $p_{21}=t^7t1$, $f_{t_2}=t^61t^61t^61$ & $[20,2,8]$\\
    \hline
    \end{longtable}
    \end{table}
  \begin{table}[H] \centering  \begin{longtable}{|c|c|c|} \caption{\label{Table3}The $2$-generator SGQC codes over $F_4+vF_4$}\\
   \hline
    $t_1, t_2, N$ & Generator polynomials  & $\phi(\mathsf{C})$\\
    \hline
    $4, 8, 12$ & $f_{t_1}=101$, $~p_{21}=101$,
    $f_{t_2}=1010101$ & $[24,4,4]$ \\
    \hline
    $4,8,12$ & $f_{t_1}=t1t1$, $~p_{21}=t^21$, $f_{t_2}=11$ & $[24,8,6]$ \\
    \hline
    $4,1,5$ & $f_{t_1}=(v+t)1(v+t)1$, $~p_{21}=(v+t)(v+t)1$, $f_{t_2}=1$& $[10,2,8]**$\\
    \hline
    $8,4,12$ & $f_{t_1}=(v+t)1(v+t)1$, $p_{21}=(v+t)1(v+t)$, $f_{t_2}=1111$& $[24,6,7]$\\
    \hline
    $8,4,12$ & $f_{t_1}=(v+t)1(v+t)1$, $p_{21}=(v+t)(v+t)1$, $f_{t_2}=1$ &$[24,9,4]$\\
    \hline
    $8,1,9$& $f_{t_1}=(v+t)1(v+t)1$, $p_{21}=(v+t)(v+t)1$, $f_{t_2}=1$ &$[18,6,7]$\\
    \hline
    $4,2,6$ & $f_{t_1}=(v+t)1(v+t)1$ c,$p_{21}=(v+t)(v+t)1$, $f_{t_2}=11$ &$[12,2,8]*$\\
    \hline
    $6,1,7$ & $f_{t_1}=1001$, $p_{21}=111$, $f_{t_2}=1$ & $[14,4,4]$\\
    \hline
\end{longtable}
\end{table}
\begin{tablenotes}
      \small
      \item\begin{center} $**$ denotes the optimal code, and $*$ denotes the near-optimal code in the table.
      \end{center}
\end{tablenotes}

\section{Conclusion}
This work studies the structure of skew generalized quasi-cyclic codes over $S$ without any restriction on length. It derives $1$-generator and multi-generator polynomial codes along with their corresponding dimensions. Furthermore, it examines the $1$-generator idempotent polynomial over $F_q+vF_q$ and $S$, provides examples, and derives linear codes. Moreover, the study establishes a lower bound on the minimum Hamming distance for the $1$-generator skew generalized quasi-cyclic codes.\par
%In future, it will be interesting to determine
In the future, one may determine the minimum distances and generator polynomials of $\mathsf{C}^\perp$ in terms of the generator polynomial of $\mathsf{C}$ for the $\rho$-generator polynomial codes. Furthermore, the construction of quantum codes based on these codes is a promising area of investigation.
\par

\section*{Acknowledgements}
The first author is thankful to the University Grants Commission (UGC), Govt. of India, for financial support under Reference No.: $221610001198$. All authors are thankful to the Indian Institute of Technology Patna for providing the research facilities. %The authors would also like to thank the Editor and anonymous referee(s) for providing constructive suggestions to improve the presentation of the manuscript.

\section{Declarations}
\textbf{Data Availability Statement}:
The authors declare that [the/all other] data supporting the findings of this study are available in this article. Any clarification may be requested from the corresponding author, provided it is essential. \\
\textbf{Competing interests}:
The authors declare that there is no conflict of interest for the publication of this manuscript.\\
\textbf{Use of AI tools declaration}:
The authors declare they did not use artificial intelligence (AI) tools to create this manuscript.


\begin{thebibliography}{99}
 \bibitem{Abualrub18}T. Abualrub,  M.F. Ezerman,  P. Seneviratne, P.  Sol\'e,  Skew generalized quasi-cyclic codes. TWMS J. Pure Appl. Math. \textbf{9}(2), 123-134 (2018)
    %DOI: https://hdl.handle.net/10356/139559
    %P.1.07
 \bibitem{Abualrub10} T. Abualrub, A. Ghrayeb, N. Aydin, I. Siap,  On the construction of skew quasi-cyclic codes. IEEE Trans. Inform. Theory \textbf{56}(5), 2081-2090 (2010)
    %DOI: 10.1109/TIT.2010.2044062
    %P.1.01
 \bibitem{Seneviratne10}T. Abualrub, P. Seneviratne, Skew codes over rings.  Proc. in IMECS, Hong Kong, \textbf{II} (2010)
    %DOI:
\bibitem{Abualrub07}T. Abualrub, I. Siap, Cyclic codes over the rings $Z_2+ uZ_ 2$ and $Z_ 2+ uZ_ 2+ u 2 Z_ 2$. Des. Codes Cryptogr. \textbf{42}, 273-287 (2007)
\bibitem{Bhaintwal12} M. Bhaintwal, Skew quasi-cyclic codes over Galois rings. Des. Codes Cryptogr. \textbf{62}, 85-101 (2012)
%DOI: https://doi.org/10.1007/s10623-011-9494-0
%P.1.30
\bibitem{Boucher07} D. Boucher, W. Geiselman, F. Ulmer, Skew cyclic codes. Appl. Algebra Engrg. Comm. Comput. \textbf{18}, 379-389 (2007)
    %DOI:https://doi.org/10.1007/s00200-007-0043-z
    %P.1.03
%DOI: https://doi.org/10.1007/s10623-006-9034-5
 \bibitem{Boucher08}D. Boucher, P. Sol\'e, F. Ulmer, Skew constacyclic codes over Galois rings. Adv. Math. Commun. \textbf{2}(3), 273-292 (2008)
   %P.1.23
\bibitem{Bonnecaze99}A. Bonnecaze, P. Udaya, Cyclic codes and self-dual codes over $F_ 2+ uF_2$, IEEE Trans. Inform. Theory \textbf{45}(4), 1250-1255 (1999)
%DOI: 10.1109/18.761278
%P.1.27
\bibitem{Dougherty99}S.T. Dougherty, P.  Gaborit, M. Harada, A. Munemasa, P. Sol\'e, Type IV self-dual codes over rings. IEEE Trans. Inform. Theory \textbf{45}(7), 2345-2360 (1999)
%DOI: 10.1109/18.796375
%P.1.28
\bibitem{Esmaeili09}M. Esmaeili, S. Yari,  Generalized quasi-cyclic codes: structural properties and code construction. Appl. Algebra Engrg. Comm. Comput. \textbf{20}, 159-173 (2009)
 %DOI: https://doi.org/10.1007/s00200-009-0095-3
 %P.1.16
 %\bibitem{Valdebenito18} A. E. L. Valdebenito, and A. L. Tironi, On the dual codes of skew constacyclic codes,~\textit{Advances in Mathematics of Communications},~ $\textbf{12(4)}(2018): ~659-679$.
    %DOI: 10.3934/amc.2018039
    %P.1.11
 \bibitem{Gao16} J. Gao, L. Shen, F.W. Fu, A Chinese remainder theorem approach to skew generalized quasi-cyclic codes over finite fields. Cryptogr. Commun. \textbf{8}, 51-66 (2016)
 %DOI: https://doi.org/10.1007/s12095-015-0140-y
 %P.1.06
\bibitem{Grassl}M. Grassl,  Table of bounds on linear codes, Available at http://www.codetables.de. (Accessed 27th March 2025)
 %P.1.15
 \bibitem{Gursoy14}F. Gursoy, I. Siap, B. Yildiz,  Construction of skew cyclic codes over $\mathbb{F}_q+v\mathbb{F}_q$. Adv. Math. Commun. \textbf{8}(3), 313-322 (2014)
    %DOI: 10.3934/amc.2014.8.313
    %P.1.05
 \bibitem{Calderbank94}A.R. Hammons, Jr.,  P.V. Kumar, R. Calderbank, N.J.A. Sloane, P. Sol\'e, The $\mathbb{Z}_4$-linearity of Kerdock, Preparata, Goethals and related codes. IEEE Trans. Inf. Theory \textbf{40}(2), 301-319 (1994)
%P.1.25
 %\bibitem{Donald74}B. R. McDonald, Finite rings with identity, ~\textit{Pure and Applied Mathematics},~  $\textbf{28}(1974)$.
 %P.1.19
  \bibitem{Islam18} H. Islam, O. Prakash,  Skew cyclic codes and skew  $(\alpha_1+u\alpha_2+v\alpha_3+uv\alpha_4)$-constacyclic codes over $F_q+uF_q+uvF_q$. Int. J. Inf. Coding Theory \textbf{5}(2), 101-116 (2018)
    %DOI: 10.1504/IJICOT.2018.10016413
    %P.1.09
 \bibitem{Islam22}H. Islam, O. Prakash,  Structure of skew cyclic codes over $\mathbb{F}_p[u]/\langle u^k\rangle$. J. Discrete Math. Sci. Cryptograpghy \textbf{26}(8), 1-13 (2022)
    %DOI: https://doi.org/10.1080/09720529.2022.2026624
    \bibitem{Jitman12}S. Jitman, S. Ling, P. Udomkavanich, Skew constacyclic codes over finite chain rings. Adv. Math. Commun. \textbf{6}(1), 29-63 (2012)
   %P.1.24
    \bibitem{Ling01}S. Ling, P. Sol\'e,   On the algebraic structure of quasi-cyclic codes I: Finite fields. IEEE Trans. Inform. Theory \textbf{47}(7), 2751-2760 (2001)
   %DOI:10.1109/18.959257
   %P.1.22
   \bibitem{Ozen19}M. Ozen, N.T. Ozzaim, H. Ince, Skew quasi-cyclic codes over $\mathbb{F}_q+v\mathbb{F}_q$. J. Algebra Appl. \textbf{18}(4), 1950077, 16 pp. (2019)
 %DOI: https://doi.org/10.1142/S0219498819500774
 %P.1.08
\bibitem{Pless97}V. Pless, P. Sol\'e, Z. Qian, Cyclic Self-Dual $Z_4$-Codes. Finite Fields Appl. \textbf{3}(1), 48-69 (1997)
%DOI: https://doi.org/10.1006/ffta.1996.0172
\bibitem{Prakash21}O. Prakash, H. Islam, S. Patel, P. Sol\'e,  New quantum codes from skew constacyclic codes over a class of non-chain rings $R_{e,q}$. Internat. J. Theoret. Phys. \textbf{60}(9), 3334-3352 (2021)
 %DOI: 10.1007/s10773-021-04910-0
 %P.1.13
  \bibitem{Prakash23}O. Prakash, A. Singh, R.K. Verma, P. Sol\'e, W. Cheng, DNA Code from Cyclic and Skew Cyclic Codes over $F_4[v]/\langle v^ 3\rangle$. Entropy \textbf{25}(2), 239 (2023)
   %DOI:  https://doi.org/10.3390/e25020239
   %P.1.18
\bibitem{Seguin04}G.E. S\'eguin,  A class of 1-generator quasi-cyclic codes. IEEE Trans. Inform. Theory \textbf{50}(8), 1745-1753 (2004)
   %DOI: 10.1109/TIT.2004.831861
   %P.1.21
   \bibitem{Seneviratne22}P. Seneviratne, T. Abualrub,  New linear codes derived from skew generalized quasi-cyclic codes of any length. Discrete Math. \textbf{345}(11), 113018 (2022)
%DOI: https://doi.org/10.1016/j.disc.2022.113018
%P.1.10
   \bibitem{Siap11}I. Siap, T. Abualrub, N. Aydin, P. Seneviratne,  Skew cyclic codes of arbitrary length. Int. J. Inf. Coding Theory \textbf{2}(1), 10-20 (2011)
    %DOI: https://doi.org/10.1504/IJICOT.2011.044674
    %P.1.04

 \bibitem{Siap05}I. Siap, N. Kulhan,  The structure of generalized quasi-cyclic codes. Appl. Math. E-Notes [electronic only]  \textbf{5}, 24-30 (2005)
    %P.1.17
%\bibitem{Ore33} O., Ore, Theory of non-commutative polynomials,~\textit{Annals of mathematics},~  $\textbf{34(3)} (1933):~ 480-508$.
%P.1.20
 \bibitem{Zhu10}S. Zhu, Y. Wang, M. Shi,  Some Results on Cyclic Codes Over ${F} _ {2}+ v {F} _ {2} $. IEEE Trans. Inform. Theory \textbf{56}(4), 1680-1684 (2010)
    %DOI: 10.1109/TIT.2010.2040896
    %P.1.29





\end{thebibliography}
\end{document}